\documentclass[10pt,twoside,leqno]{article}
\usepackage[letterpaper]{geometry}

\usepackage{iftex}
\ifXeTeX
  \usepackage{scrlfile}
  \PreventPackageFromLoading{breakurl}
  \PreventPackageFromLoading{hypdvips}
\fi
\usepackage{siamproceedings}
\newsiamthm{claim}{Claim}
\newsiamremark{remark}{Remark}
\usepackage[T1]{fontenc}
\usepackage{amssymb,mathtools}
\usepackage{algpseudocode,booktabs}
\usepackage{needspace}
\usepackage{microtype}
\usepackage{enumitem}
\newcommand{\opt}{\mathrm{opt}}

\usepackage{xcolor}

\title{\Large Improved Approximation for Unsplittable Capacitated Vehicle Routing Using Quota Trees}
\author{Zhenyu Wang\thanks{Southern University of Science and Technology
(\texttt{\email{wangzy20233@mail.sustech.edu.cn}).}}}
\date{}
\begin{document}
\maketitle

\begin{abstract}
In the \emph{Unsplittable Capacitated Vehicle Routing Problem (UCVRP)},
we are given a set $V$ of $n$ clients with normalized demands
$d_v \in (0,1]$, a depot $r$, a unit vehicle capacity, and symmetric
metric travel costs $c$ on $V \cup \{r\}$. We seek a collection of routes
of minimum total cost that start and end at $r$, serve each client's
entire demand on a single route, and respect the unit capacity constraint
$\sum_{v \in V(R)} d_v \le 1$ on every route $R$, where $V(R)$ denotes
the clients served by $R$.

UCVRP generalizes the metric \emph{traveling salesman problem (TSP)}
and is therefore NP-hard. Capacity constraints add a packing
difficulty: when all clients are colocated away from the depot,
minimizing the route cost is equivalent to minimizing the number of
bins in \emph{bin packing}. Consequently, UCVRP admits no polynomial-time
approximation ratio below $3/2$ unless \textbf{P=NP}~\cite{GrandoniMathieuZhou23}.
The classical $3.5$-approximation has been improved; Zhao and
Xiao~\cite{ZhaoXiao26Stochastic} recently obtained a
$(2+\ln 2+\varepsilon)$-approximation.

In this paper, we give a deterministic polynomial algorithm with
approximation ratio $5/2$ for general capacity and arbitrary positive demands.
Our algorithm uses local search to select routes through one client
with large demand and clients with small demands.
It constructs candidate routes using endpoint quota trees after rounding
demands and creating vertex copies.
We replace selected routes using sets of two to four clients and a
squared-weight local search, keeping their client sets disjoint.
The remaining clients are served by partitioning a tour constructed
using weighted graphic matroid parity.
We also give a deterministic polynomial $5/2$-approximation algorithm
for the undirected capacitated arc routing problem (CARP) with arbitrary
positive edge demands.
\end{abstract}

\section{Introduction}
In the \emph{Unsplittable Capacitated Vehicle Routing Problem (UCVRP)},
the input consists of a depot $r$, a client set $V$ with $n=|V|$, and a
symmetric metric $c$ on $V\cup\{r\}$. Each client $v\in V$ has a
demand $d_v\in(0,1]$, where the vehicle capacity is normalized to one.
A \emph{feasible solution} is a collection $\mathcal T$ of routes
that start and end at $r$. For each route $R$, let $V(R)$ be the set
of clients it serves and $c(R)$ its cost. The sets $V(R)$,
$R\in\mathcal T$, must partition $V$, and
$\sum_{v\in V(R)}d_v\le1$ for every $R\in\mathcal T$.
Each client's demand is served in full by its assigned route.
Our goal is to minimize $c(\mathcal T):=\sum_{R\in\mathcal T}c(R)$;
we denote this minimum by $\opt$. This models deliveries from a
common depot using vehicles of limited capacity, with each client's
order delivered on a single trip.

The study of \emph{Capacitated Vehicle Routing (CVRP)} dates back
to the work of Dantzig and
Ramser~\cite{DantzigRamser59} on truck dispatching in 1959.
UCVRP contains the metric \emph{Traveling Salesman Problem (TSP)}
as a special case: when $\sum_{v\in V}d_v\le1$, an optimal solution
can be chosen to consist of a single tour. Thus UCVRP is NP-hard,
which motivates the study of approximation algorithms.

A classical approach, known as \emph{tour partitioning}, is to
compute a TSP tour $R_{\mathrm{TSP}}$ through the depot and all clients,
then partition its client order into consecutive sets of total
demand at most one. Connecting the ends of each set to the depot
gives feasible routes. This approach was studied by
Haimovich and Rinnooy Kan~\cite{HaimovichRinnooyKan85} and extended
to unequal demands by Altinkemer and Gavish~\cite{AltinkemerGavish87}.
Let $\alpha$ be the approximation ratio of the algorithm used to
compute the TSP tour.
Karlin, Klein, and Oveis Gharan~\cite{KarlinKO24,KarlinKO23}
obtained a ratio below $3/2$, including a deterministic algorithm.
Building on their work, Gurvits, Klein, and Leake~\cite{GurvitsKL24}
improved the guarantee to $\alpha\le3/2-2.18\cdot10^{-34}$, also
achievable deterministically. We state the numerical bounds below
using $\alpha=3/2$; the bounds depending on $\alpha$ admit the
corresponding slight improvement.
The shortest such tour costs at most $\opt$:
concatenate the routes of an optimal solution at the depot, retain
the initial and final depot visits and the first visit to each client,
and connect consecutive retained vertices by direct edges.
The triangle inequality ensures that the cost does not increase.
Hence $c(R_{\mathrm{TSP}})\le\alpha\opt$.
The classical partitioning algorithm produces routes of total cost
at most $c(R_{\mathrm{TSP}})+2\opt\le(\alpha+2)\opt$.
With Christofides' algorithm~\cite{Christofides22}, $\alpha=3/2$,
giving a $3.5$-approximation.

Blauth, Traub, and Vygen~\cite{BlauthTraubVygen23} improved this
guarantee by exploiting the structure of instances on which the
classical analysis is nearly tight. Friggstad, Mousavi, Rahgoshay,
and Salavatipour~\cite{FriggstadMRS25} combined matching with an
improved tour-partitioning procedure to obtain an
$(\alpha+1.75)$-approximation. They also used rounding of a
\emph{configuration linear program (LP)} to serve clients with
large demands before partitioning
a tour through the remaining clients. This gives an
$(\alpha+1+\ln 2+\varepsilon)$-approximation for every fixed
$\varepsilon>0$. Taking $\alpha=3/2$ and sufficiently small
$\varepsilon$ gives a $3.194$-approximation.
Combining their analysis with that of Blauth
et al. yields a further constant improvement.

More recently, Ebert and Weismantel~\cite{EbertWeismantel26} introduced
an \emph{Average Greedy} algorithm that controls both the cost of selected
routes and their coverage of clients with large demands. Combining
this algorithm with matching and tour partitioning gives an
$(\alpha+1.659)$-approximation, and thus a $3.159$-approximation
for symmetric metrics. Zhao and Xiao~\cite{ZhaoXiao26} combined matching,
LP rounding, and tour partitioning. Their revised analysis, based on
demand discretization and factor-revealing LPs, gives expected and
deterministic approximation ratios of about $3.147$ and
$3.148$, respectively, for general capacity.

Zhao and Xiao~\cite{ZhaoXiao26Stochastic} gave a new tour construction
using \emph{weighted graphic matroid parity}. Combining this construction
with the LP rounding of Friggstad et al.~\cite{FriggstadMRS25} gives a
$(2+\ln 2+\varepsilon)$-approximation for UCVRP.
We introduce a local-search algorithm based on quota trees.
Combined with set replacements and the tour construction of Zhao
and Xiao, it gives the following guarantee.

\begin{theorem}[Main result]\label{thm:main}
Algorithm~\ref{alg:global} is a deterministic
$5/2$-approximation for metric UCVRP.
Its running time is
\begin{equation}\label{eq:runtime}
n^{O(1)}.
\end{equation}
\end{theorem}

We also consider the \emph{undirected capacitated arc routing problem}
(CARP), in which each required edge is traversed and its entire demand
is served by one route starting and ending at a common depot.
\begin{theorem}\label{thm:carp}
The algorithm in Appendix~\ref{app:carp} is a deterministic
$5/2$-approximation for CARP.
Its running time is polynomial.
\end{theorem}

\subsection{Related work}
In unit-demand CVRP, all clients have normalized demand $d_v=1/k$
for an integer $k\ge1$, so every route serves at most $k$ clients.
In \emph{splittable CVRP}, a client's demand may be divided among
several routes. For both variants, classical tour partitioning gives
an $(\alpha+1)$-approximation~\cite{HaimovichRinnooyKan85,BlauthTraubVygen23}.
Blauth, Traub, and Vygen~\cite{BlauthTraubVygen23} improved this
guarantee to $5/2-1/3000$ when $\alpha=3/2$.
The tour construction of Zhao and Xiao~\cite{ZhaoXiao26Stochastic}
gives a $2$-approximation for both variants.

For Euclidean distances, Das and Mathieu~\cite{DasMathieu15} gave
a QPTAS for
unit-demand CVRP in any fixed dimension. For every fixed
$\varepsilon>0$, their algorithm achieves a $(1+\varepsilon)$-approximation
in $n^{\operatorname{poly}(\log n)}$ time.
For unit-demand CVRP in the Euclidean plane with arbitrary capacity,
Sitters~\cite{Sitters25} improved the running time to
$n^{f(\varepsilon)\log\log n}$, where $f$ depends only on $\varepsilon$.
In the Euclidean plane, Friggstad, Grandoni, and
Mousavi~\cite{FriggstadGrandoniMousavi26} obtained a polynomial-time
$(2-\tau)$-approximation for unit-demand CVRP, for an absolute
constant $\tau>10^{-5}$. For UCVRP in the Euclidean plane,
Grandoni, Mathieu, and Zhou~\cite{GrandoniMathieuZhou23} gave a
polynomial-time $(2+\varepsilon)$-approximation for every fixed
$\varepsilon>0$.

In tree metrics, distances are lengths of paths in an edge-weighted
tree. Mathieu and Zhou~\cite{MathieuZhou23PTAS} gave a
PTAS for unit-demand
CVRP on trees with arbitrary capacity: for every fixed
$\varepsilon>0$, the algorithm achieves a $(1+\varepsilon)$-approximation
in polynomial time. For unsplittable demands, Mathieu and
Zhou~\cite{MathieuZhou23} obtained a $(3/2+\varepsilon)$-approximation.
Approximating UCVRP on trees within a factor below $3/2$ is NP-hard,
so this guarantee is tight up to $\varepsilon$ unless \textbf{P=NP}.

Approximation schemes are also known for other graph metrics.
Becker, Klein, and Schild~\cite{BeckerKS19} gave a PTAS for
unit-demand CVRP in planar graph metrics when the capacity $k$ is
fixed. Jayaprakash and Salavatipour~\cite{JayaprakashS23} gave
QPTASes for unit-demand CVRP with unrestricted capacity on graphs
of bounded \emph{treewidth} or \emph{highway dimension}, and in
metrics of bounded \emph{doubling dimension}.
For \emph{graphic metrics}, induced by shortest-path distances in
connected unweighted undirected graphs, M{\"o}mke and
Zhou~\cite{MomkeZ23} obtained a $1.95$-approximation for unit-demand
CVRP when every vertex other than the depot is a client.

In the \emph{fixed-capacity} setting, the vehicle capacity $k$ is
a fixed positive integer and demands are positive integers before
normalization. For UCVRP, this means
$d_v\in\{1/k,2/k,\ldots,1\}$ and hence $|V(R)|\le k$ for every
feasible route $R$. Feasible routes can therefore be enumerated in
$n^{O(k)}$ time. Zhao and Xiao~\cite{ZhaoXiaoFixed26} studied this
setting in general metrics. Using $\alpha=3/2$, they obtained
approximation ratios $5/2-\Theta(1/\sqrt{k})$ for unit-demand and
splittable CVRP, and $5/2+\ln 2-\Theta(1/\sqrt{k})$ for UCVRP.

The \emph{multidepot CVRP} allows several depots and requires the
solution to choose a depot for each route. In the
\emph{fixed-destination} variant, each route returns to the depot
from which it starts. For an arbitrary number of depots in general
metrics, Zhao and Xiao~\cite{ZhaoXiaoMultidepot25} obtained
approximation ratios $4-1/1500$ for the unit-demand and splittable
variants, and $4-1/50000$ for unsplittable demands.
Friggstad and M{\"o}mke~\cite{FriggstadMomke26} obtained a
$3.9365$-approximation for the unit-demand variant using LP rounding.
More recently, Zhao and Xiao~\cite{ZhaoXiao26Stochastic} obtained
approximation ratios of $10/3$ for unit-demand and splittable CVRP,
and $3.792$ for unsplittable demands.

The undirected capacitated arc routing problem (CARP) was introduced
by Golden and Wong~\cite{GoldenWong81}. Eiselt, Gendreau, and
Laporte~\cite{EiseltGL95} surveyed the rural postman problem and its
extensions, including CARP. Jansen~\cite{Jansen93} developed
tour-partitioning algorithms for metric CARP, and
W{\o}hlk~\cite{Wohlk08} later gave a dynamic-programming approach
with the same worst-case guarantees. For general demands, these
algorithms give a capacity-independent approximation ratio of $7/2$.
An approximation-preserving reduction of van Bevern et
al.~\cite{vanBevernHNS14} extends metric CARP guarantees to arbitrary
nonnegative edge costs.

For unit-demand CARP with vehicle capacity $k$, Jansen~\cite{Jansen93}
obtained an approximation ratio of $5/2-3/(2k)$. Zhao and
Xiao~\cite{ZhaoXiaoCARP25} improved this bound to
$5/2-\Theta(1/\sqrt{k})$. More recently, Zhao and
Xiao~\cite{ZhaoXiao26Stochastic} gave a $2$-approximation for
unit-demand CARP, uniformly over all capacities. For general
unsplittable edge demands, they gave a $3$-approximation and stated
a further improvement to $2+\ln2+\varepsilon$ using LP rounding.
Our deterministic polynomial algorithm improves the general-demand
guarantee to $5/2$ and allows arbitrary positive edge demands.

\subsection{Our approach}
We use the \emph{radial lower bound}
$D:=2\sum_{v\in V}d_vc(r,v)\le\opt$.
Our algorithm first selects some disjoint feasible routes. It serves
the remaining clients by partitioning the order of a TSP tour into
consecutive sets of total demand at most one.
Each set becomes a route by adding an edge from and to the depot.
We compute the tour on all clients using the algorithm of Zhao and
Xiao~\cite{ZhaoXiao26Stochastic}. For an optimal solution $\mathcal T^*$,
write $\rho(R)=\max_{v\in V(R)}c(r,v)$ for nonempty $R$ and set
$\rho(\varnothing)=0$. Their construction gives
\[
c(R_{\mathrm{TSP}})\le
\sum_{R\in\mathcal T^*}\bigl(2c(R)-2\rho(R)\bigr).
\]
For a subset $X$ of a tour's clients, the \emph{restricted tour} retains
the initial and final depot visits and the clients of $X$ in their
original order, connecting consecutive retained vertices by direct edges.
Its cost does not increase by the triangle inequality.
After route selection, we restrict $R_{\mathrm{TSP}}$ to the unserved clients.

The $\delta$-tank theorem of Friggstad et
al.~\cite{FriggstadMRS25}, used with $\delta=1/8$,
gives a bound for any client set $X$.
Serve each client $v\in X$ with $d_v>1/2$ by $(r,v,r)$, restrict
the tour to the other clients, and partition the resulting order.
These operations give a cost of at most the TSP cost plus $p(X)$, where
$p(X)=\sum_{v\in X}p_v$ and
\[
p_v=\min\left\{2c(r,v),\,
  \frac{16}{7}\left[d_v+(d_v-\tfrac18)_+\right]c(r,v)\right\},
\qquad (x)_+=\max\{x,0\}.
\]
We call $p_v$ the price of client $v$.
A client is small if $d_v\le1/8$, and big if $d_v>7/16$.
A client is non-small if $d_v>1/8$.
In particular, a small client has price $p_v=(16/7)d_vc(r,v)$.

For a selected route family $\mathcal R$, write $V(\mathcal R)$ for
its served clients and $c(\mathcal R)$ for its total cost.
Completion gives a solution of total cost at most
$c(R_{\mathrm{TSP}})+c(\mathcal R)+p(V\setminus V(\mathcal R))$.
To obtain a $5/2$-approximation, we bound
$c(\mathcal R)+p(V\setminus V(\mathcal R))$ by
$\sum_{R\in\mathcal T^*}(2\rho(R)+c(R)/2)$.
Adding the tour bound then gives the claimed guarantee.
For a route $R$, the saving $s(R):=p(V(R))-c(R)$ is obtained by
serving $V(R)$ on $R$ instead of keeping all their prices in the bound.
The identity $c(\mathcal R)+p(V\setminus V(\mathcal R))
=p(V)-\sum_{R\in\mathcal R}s(R)$ turns route selection into the
problem of obtaining enough total saving.

Set packing selects feasible sets of two, three, or four non-small
clients, each served by a cheapest route.
The squared-weight local-search analysis of Berman~\cite{Berman00},
with weight rounding, gives comparison coefficients
$3/5$, $9/20$, and $9/25$ for sets of size two, three, and four.
We give the algorithm and its proof in
Theorem~\ref{thm:setpacking} and Appendix~\ref{app:proof-setpacking}.
Lemma~\ref{lem:exceptional-range} shows that these coefficients give
the required routewise bound except on routes with exactly one
non-small client.
Call such a route $R$ exceptional if
$\Delta(R):=p(V(R))-2\rho(R)-c(R)/2>0$, in which case its unique
non-small client $v$ satisfies $7/16<d_v<9/16$ and $c(r,v)>0$.

Local search maintains, for each big client $v$, either an empty route
or a route $R_v$ through $v$ and a set $S_v$ of small clients.
The nonempty routes have disjoint client sets.
For clients without a nonempty stored route, set
$R_v=S_v=\varnothing$ and $c(R_v)=s(R_v)=0$.
A replacement for $R_v$ serves $v$ and small clients that belong
to $S_v$ or are unassigned.
To account for this restriction, we use the following score for each
nonempty stored route:
\[
\Phi(R_v)=p_v+\frac35p(S_v)-c(R_v),
\qquad
s(R_v)=\Phi(R_v)+\frac25p(S_v),
\]
with $\Phi(\varnothing)=0$.
Thus the saving exceeds the score by $(2/5)p(S_v)$.

Write $h_v=2c(r,v)$ and fix a rational constant
$0<\varepsilon_0\le10^{-3}$, independent of the input.
When constructing a replacement for $R_v$, use only available small
clients at distance at least $\tau_v=6c(r,v)/7$ from the depot.
Every available small client $u$ with $c(r,u)\ge\tau_v$ satisfies
$p_u=(16/7)d_uc(r,u)\ge(16/7)\tau_vd_u$.
We construct candidates using the endpoint quota-tree theorem of
Chaudhuri et al.~\cite{ChaudhuriGRT03}, in the form stated by
Chekuri et al.~\cite[Theorem~3.3]{ChekuriKP12}.

To make the vertex quota express demand, we replace clients by
labelled copies at their original locations.
We enumerate the small clients of demand greater than
$\varepsilon_0/4$ that a candidate must serve, and use copies
to represent rounded demands of the other clients.
A greedy scan selects whole clients within capacity.
Let $P$ be an $r$--$v$ path, with total edge cost $c(P)$, whose
available small clients above the cutoff have total demand at most $1-d_v$.
Lemma~\ref{lem:demand-recovery} gives a tested route of cost at most
$2(1+\varepsilon_0/4)c(P)$ whose small-client demand is at least the
total demand of those clients minus $\varepsilon_0/4$.
The route is obtained by doubling tree edges outside the $r$--$v$
path, following an Euler trail, and adding $(v,r)$.
It retains the first visit to each selected small client and the
terminal visit to $v$, connecting consecutive retained vertices
directly. Section~\ref{sec:quota} gives the construction.

For an exceptional optimal route $R_v^*$, let $U_v$ be its small
clients currently served by stored routes other than $R_v$.
The two $r$--$v$ paths of $R_v^*$ have total cost $c(R_v^*)$.
Apply demand recovery to both paths and average the resulting score
bounds. For the tested route $R'$ of maximum score, $\Phi(R')$ is at least
this average.
The factor $1+\varepsilon_0/4$ increases the cost bound, while the
loss of up to $\varepsilon_0/4$ demand reduces the price term.
Together they decrease the averaged score bound by at most
$(\varepsilon_0/2)c(R_v^*)$. The cutoff inequality and the capacity
constraint give a positive term of at least $\rho(R_v^*)/140$.
Since $c(R_v^*)<16\rho(R_v^*)/7$, at least
$6\varepsilon_0\rho(R_v^*)$ remains after these losses.
Lemma~\ref{lem:route-score} therefore gives
$\Phi(R')\ge2\Delta(R_v^*)-(3/10)p(U_v)+6\varepsilon_0\rho(R_v^*)$.

Local search accepts $R'$ only if
$\Phi(R')>\Phi(R_v)+\varepsilon_0h_v$.
The normalized sum $\sum_{v:h_v>0}\Phi(R_v)/h_v$ starts at zero,
stays at most $n$, and increases by more than $\varepsilon_0$ at each update.
There are therefore fewer than $n/\varepsilon_0$ updates.
At termination, the final complete scan gives
\[
\Phi(R_v)\ge2\Delta(R_v^*)-\frac3{10}p(U_v)+4\varepsilon_0\rho(R_v^*).
\]
Here the loss from the stopping condition satisfies
$\varepsilon_0h_v\le2\varepsilon_0\rho(R_v^*)$.
Each small client occurs in at most one $U_v$.

After local search, set packing uses the adjusted prices
$p'_v=p_v-s(R_v)$.
The weight of a feasible non-small set $S$ is
$p'(S)-c(R_S)$, where $R_S$ is its cheapest route.
Thus the weight deducts the actual savings of the stored routes
whose big clients belong to $S$.
For selected sets, delete these stored routes and return their small
clients to the set served by tour partitioning.
Keep every stored route whose big client is not selected.
Lemma~\ref{lem:replacement} shows that this operation is accounted
for exactly by the adjusted prices.

The packing bound decreases by the full price reduction on
exceptional routes, and by at least $2/5$ of the reductions on other
routes. Substituting
$s(R_v)=\Phi(R_v)+(2/5)p(S_v)$ gives small-client terms of at least
$(4/25)\sum_v p(S_v)$.
For each exceptional route $R_v^*$, the stored-score bound gives
$\Delta(R_v^*)\le\Phi(R_v)/2+(3/20)p(U_v)$.
The sum of the unavailable-client terms is at most
$(3/20)\sum_v p(S_v)$.
Since $4/25-3/20=1/100$, these losses are covered.
Lemma~\ref{lem:selection} gives the required bound on selected-route
costs and unserved-client prices, and the TSP bound completes the
$5/2$ guarantee.

For fixed $\varepsilon_0$, candidate construction, both local searches,
the tour construction, and the final dynamic program run in polynomial time.

\Needspace{8\baselineskip}
\section{Preliminaries}\label{sec:prelim}
\subsection{Problem and notation}
\begin{definition}[\textsc{Unsplittable Capacitated Vehicle Routing}]\label{def:cvrp}
An instance $(V,r,c,d)$ of \textsc{Unsplittable Capacitated Vehicle Routing} (UCVRP) consists of:
\begin{itemize}[nosep]
    \item a set of clients $V$, where $|V| = n$,
    \item a depot $r$, not in $V$,
    \item symmetric metric travel costs $c : (V \cup \{r\}) \times (V \cup \{r\}) \to \mathbb{R}_{\ge 0}$,
    \item a demand $d_v \in (0, 1]$ for each client $v \in V$.
\end{itemize}
A feasible solution is a collection of routes $\mathcal{T}$ such that
\begin{itemize}[nosep]
    \item every route $R \in \mathcal{T}$ starts and ends at $r$,
    \item every client belongs to exactly one route,
    \item $\sum_{v\in V(R)}d_v\le1$ for all $R\in\mathcal T$, where $V(R)$ is the set of clients served by $R$.
\end{itemize}
The cost $c(R)$ of a route is the sum of its edge costs, counting
traversals with multiplicity. The objective is to minimize
$c(\mathcal T)=\sum_{R\in\mathcal T}c(R)$.
Fix an optimal solution $\mathcal T^*$ and write $\opt=c(\mathcal T^*)$.
\end{definition}

We assume that demands and distances are rational.
We use $R$ for a route, $P$ for a path, and $T$ for a tree;
$c(P)$ and $c(T)$ denote the sums of their edge costs.
For $S\subseteq V$, write $d(S)=\sum_{v\in S}d_v$.
For a nonempty route $R$, write
$\rho(R)=\max_{v\in V(R)}c(r,v)$, and set $\rho(\varnothing)=0$.
For each client $v$, write $h_v=2c(r,v)$.
For a route family $\mathcal R$, write
$V(\mathcal R)=\bigcup_{R\in\mathcal R}V(R)$ and
$c(\mathcal R)=\sum_{R\in\mathcal R}c(R)$.
The family is disjoint if no client belongs to two of its routes.

Split each route at its intermediate depot visits, and retain only one
visit to each client across the resulting routes. Connecting consecutive
retained vertices directly preserves feasibility and does not increase
the total cost, by the triangle inequality.
We therefore assume that each reference route visits each of its
clients once and has no intermediate depot visit.

\begin{definition}\label{def:classes}
Clients in $V_{\mathrm{big}}=\{v\in V:d_v>7/16\}$ are called
\emph{big}, and clients in $V_{\mathrm{small}}=\{v\in V:d_v\le1/8\}$
are called \emph{small}. A client is non-small if $d_v>1/8$.
\end{definition}

For each client $v$ on a feasible route $R$, the two parts of $R$
connecting $r$ and $v$ each cost at least $c(r,v)$.
Thus $2c(r,v)\le c(R)$. Multiplying by $d_v$, summing over the
route, and using $d(V(R))\le1$ gives the first line below.
Summing over $\mathcal T^*$ gives the radial lower bound $D$:
\begin{equation}\label{eq:radial}
\begin{aligned}
2\sum_{v\in V(R)}d_vc(r,v)&\le d(V(R))c(R)\le c(R),\\
D:=2\sum_{v\in V}d_vc(r,v)&\le\sum_{R\in\mathcal T^*}c(R)=\opt.
\end{aligned}
\end{equation}

We use the tour construction of Zhao and Xiao. Its bound uses
the maximum distance to the depot on each reference route.

\begin{theorem}[Zhao and Xiao~{\cite[Lemmas~4, 6, and~7]{ZhaoXiao26Stochastic}}]\label{thm:full-tour}
A deterministic algorithm constructs a TSP tour $R_{\mathrm{TSP}}$
on $V\cup\{r\}$ in $O(n^7)$ time. For every feasible solution
$\mathcal T$ with nonempty routes, it satisfies
\begin{equation}\label{eq:full-tour}
c(R_{\mathrm{TSP}})\le
\sum_{R\in\mathcal T}\bigl(2c(R)-2\rho(R)\bigr).
\end{equation}
\end{theorem}

In their analysis, Lemma~6 uses the partition
$\{V(R):R\in\mathcal T\}$ to construct a \emph{parity basis}
of weight at most the sum in \eqref{eq:full-tour}.
Lemma~4 converts a parity basis into a tour whose cost is at most the
basis weight, and Lemma~7 gives the $O(n^7)$ running time.

\subsection{Tour partitioning and prices}
The $\delta$-tank theorem bounds the cost of partitioning a tour.
Write $(x)_+=\max\{x,0\}$.

\begin{theorem}[Friggstad et al.~{\cite{FriggstadMRS25}}]\label{lem:tank}
Let $X\subseteq V$, let $R_{\mathrm{TSP}}$ be a TSP tour on
$X\cup\{r\}$, and let $0\le\delta\le1/2$.
Its client order can be partitioned into consecutive segments of
demand at most one. Joining each segment to the depot at both ends
gives a route family $\mathcal R$ satisfying
\begin{equation}\label{eq:tank}
c(\mathcal R)\le c(R_{\mathrm{TSP}})
+\frac{2}{1-\delta}\sum_{v\in X}
\bigl[d_v+(d_v-\delta)_+\bigr]c(r,v).
\end{equation}
\end{theorem}

For clients with $d_v\le1/2$, we use the charge in
Theorem~\ref{lem:tank} with $\delta=1/8$.
For $d_v>1/2$, the singleton route $(r,v,r)$ costs at most this
charge. We therefore define the client prices by
\begin{equation}\label{eq:price}
p_v=
\begin{cases}
\frac{16}{7}d_vc(r,v),&d_v\le1/8,\\[2pt]
\frac{32d_v-2}{7}c(r,v),&1/8<d_v\le1/2,\\[2pt]
2c(r,v),&d_v>1/2.
\end{cases}
\end{equation}

For any client price vector $p$, write $p(S)=\sum_{v\in S}p_v$.
The coefficients in \eqref{eq:price} lie in $[0,2]$, so
$0\le p_v\le h_v$. The same formula gives $p_v\le2h_vd_v$.

Lemma~\ref{lem:completion} combines these prices with tour partitioning:
given $\varnothing\ne X\subseteq V$ and a TSP tour $R_{\mathrm{TSP}}$ on
$X\cup\{r\}$, it constructs feasible routes of
cost at most $c(R_{\mathrm{TSP}})+p(X)$.

\subsection{Set packing and quota trees}

The route-selection step uses sets of two to four non-small clients.
We use the squared-weight local-search analysis of
Berman~\cite{Berman00}, with a bound that depends on each
comparison set's size. We give the algorithm and proof below.
For a weight function on sets, its value on a family denotes the sum
of its values on the family members.

\begin{theorem}\label{thm:setpacking}
Let $\mathcal F$ be an explicitly listed family of $N$ sets, each of
size two, three, or four, with positive rational weights $w(S)$.
A deterministic algorithm returns a pairwise disjoint subfamily
$\mathcal S\subseteq\mathcal F$ such that, for every pairwise disjoint
subfamily $\mathcal O\subseteq\mathcal F$,
\begin{equation}\label{eq:setpacking-guarantee}
\sum_{T\in\mathcal S}w(T)\ge
\frac9{10}\sum_{S\in\mathcal O}\frac{2w(S)}{|S|+1}.
\end{equation}
The algorithm has runtime $O(N^8)$.
\end{theorem}

If $N=0$, return the empty family. Otherwise, set
$q(S)=2w(S)/(|S|+1)$, let $Q=\max_{S\in\mathcal F}q(S)$, and use
the rounding unit $a=Q/(10N)$. Set $z(S)=\lfloor q(S)/a\rfloor$ and
discard sets with $z(S)=0$. Every remaining rounded weight $z(S)$ is an integer
between $1$ and $10N$.

Start with $\mathcal S=\varnothing$. For a family $\mathcal C$ of at
most four pairwise disjoint remaining sets, define
$N_{\mathcal S}(\mathcal C)=
\{T\in\mathcal S:\exists S\in\mathcal C,
T\cap S\ne\varnothing\}$.
If $\sum_{S\in\mathcal C}z(S)^2>
\sum_{T\in N_{\mathcal S}(\mathcal C)}z(T)^2$, replace
$\mathcal S$ by
$(\mathcal S\setminus N_{\mathcal S}(\mathcal C))\cup\mathcal C$.
Enumerate the families $\mathcal C$ in a fixed order and restart the
scan after every accepted exchange. Stop when a complete scan accepts
none. Members common to $\mathcal S$ and $\mathcal C$ contribute the
same amount on both sides and can be omitted from the enumeration.
The replacement preserves pairwise disjointness.

The proof in Appendix~\ref{app:proof-setpacking} bounds the number of
exchanges and proves \eqref{eq:setpacking-guarantee}. For each selected
set $T\in\mathcal S$, it uses the fact that $T$ intersects at most
$|T|$ sets in any pairwise disjoint family.

Local search requires routes through one big client and a set of small
clients. We construct a tree containing the depot and the big client.
To control the demand represented in the tree, we encode rounded client
demands by labelled vertex copies and impose a lower bound on the number
of tree vertices. We then select clients represented in the tree using
their original demands to satisfy the capacity constraint, and convert
the tree into a route. Section~\ref{sec:quota} gives this construction.
The following theorem bounds the tree's cost by the cost of a path
with the same endpoints and vertex quota.
\begin{theorem}[Chaudhuri et al.~{\cite{ChaudhuriGRT03}};
Chekuri et al.~{\cite[Theorem~3.3]{ChekuriKP12}}]\label{thm:quota-tree}
Let $G=(V,E)$ be a complete undirected graph with $N=|V|$ and
nonnegative edge costs satisfying the triangle inequality.
Distinct vertices may have zero distance.
Fix two distinct vertices $r,v\in V$, an integer quota $2\le k\le N$,
and $0<\varepsilon\le1$.
Let $L$ denote the minimum cost among $r$--$v$ paths with $k$ vertices.
A deterministic algorithm with runtime $N^{O(1/\varepsilon)}$ returns
a tree $T$ containing $r$, $v$, and at least $k$ vertices in total,
with $c(T)\le(1+\varepsilon)L$.
\end{theorem}
If the vertices at zero distance from $r$ include $v$ and at least
$k$ vertices, return a tree of cost zero on these vertices.
Otherwise, let $a$ be the minimum positive edge cost, increase each
edge cost between distinct vertices by $\varepsilon a/(4N)$, which
preserves the triangle inequality, and apply the theorem with error
$\varepsilon/2$. A minimum-cost $r$--$v$ path with $k$ vertices
has cost $L\ge a$ and at most $N-1$ edges.
The original cost of the returned tree is therefore at most
$(1+\varepsilon/2)(1+\varepsilon/4)L\le(1+\varepsilon)L$.

\Needspace{6\baselineskip}
\section{The algorithm}\label{sec:global}
\subsection{Main algorithm}

Let $R_{\mathrm{TSP}}$ be the tour on $V\cup\{r\}$ returned by
Theorem~\ref{thm:full-tour}. Let $\mathcal R$ be a feasible partial
route family with pairwise disjoint client sets, and put
$X=V\setminus V(\mathcal R)$. Applying Lemma~\ref{lem:completion}
to the restricted tour on $X$ gives total cost at most
$c(R_{\mathrm{TSP}})+c(\mathcal R)+p(X)$.
Define the saving of a route as $s(R)=p(V(R))-c(R)$, using the fixed
prices in \eqref{eq:price}. It therefore suffices to construct
$\mathcal R$ such that
\begin{equation}\label{eq:selection-goal}
c(\mathcal R)+p(X)\le
\sum_{R\in\mathcal T^*}\left[2\rho(R)+\frac12c(R)\right].
\end{equation}
Theorem~\ref{thm:full-tour} then gives a total cost of at most
$(5/2)\opt$.

For each $v\in V_{\mathrm{big}}$, local search maintains either a
feasible route $R_v$ through $v$ whose other clients are small,
or the empty route. Write $S_v=V(R_v)\setminus\{v\}$ for its small-client
set. We use $V(\varnothing)=\varnothing$ and
$c(\varnothing)=s(\varnothing)=0$; for $v\notin V_{\mathrm{big}}$,
set $R_v=S_v=\varnothing$. Local search compares routes by the score
\begin{equation}\label{eq:score}
\Phi(R)=s(R)-(1-\lambda)p(V(R)\cap V_{\mathrm{small}}),
\qquad \lambda=\frac35,\qquad \Phi(\varnothing)=0.
\end{equation}
Thus a nonempty stored route has score
$\Phi(R_v)=p_v+\lambda p(S_v)-c(R_v)$.
We use the fixed internal accuracy $\varepsilon_0$ from the introduction.
All visit orders and choices among tied candidates use fixed orders.

For a big client $v$ with $h_v=2c(r,v)>0$, let
$A\subseteq V_{\mathrm{small}}$ be the set of small clients not served
by other current routes. On input $(v,A)$,
Algorithm~\ref{alg:quota} uses only clients at distance at least
$6c(r,v)/7$ from the depot and returns a feasible route serving $v$.
Section~\ref{sec:quota} gives the construction,
Lemma~\ref{lem:route-score} bounds the returned route's score, and
Appendix~\ref{sec:runtime} analyzes the running time.
A replacement is accepted only if its score exceeds the stored score
by more than $\varepsilon_0h_v$.

Lemma~\ref{lem:local} guarantees client-disjoint stored routes,
nonnegative scores, and $0\le s(R_v)\le p_v$.
It also bounds the number of accepted updates by fewer than $n/\varepsilon_0$.
After local search, $R_v,S_v$ denote its final stored routes and
client sets. These values, and their savings, remain fixed during
the set replacements, including for routes omitted from the solution.
Set packing uses the adjusted prices $p'_v=p_v-s(R_v)$,
which satisfy $0\le p'_v\le p_v$ by the lemma. The saving $s$ and score
$\Phi$ continue to use the original prices $p$.

For a feasible set $S$ of two to four non-small clients, $R_S$ denotes
a cheapest route visiting exactly $S$. The set weight $w(S)$ deducts
the savings of the stored routes replaced by $R_S$.
Section~\ref{subsec:packing-construction} defines these weights and the
resulting route family. Section~\ref{subsec:dp} gives the recurrence
for the final partition.

\begin{algorithm}[!htb]
\caption{$5/2$-Approximation for Unsplittable CVRP}
\label{alg:global}
\begin{algorithmic}[1]
\Require Instance $(V,r,c,d)$.
\Ensure A feasible route family serving every client exactly once,
with total cost at most $(5/2)\opt$.
\State Compute $D=2\sum_{v\in V}d_vc(r,v)$
\If{$|V|\le1$ or $D=0$}
  \State \Return singleton routes for all clients
\EndIf
\State Apply Theorem~\ref{thm:full-tour} to obtain $R_{\mathrm{TSP}}$ on $V\cup\{r\}$
\State Compute the fixed prices $p$ by \eqref{eq:price}
\State Set $R_v\gets(r,v,r)$ if $d_v>1/2$ and $h_v>0$; otherwise $R_v\gets\varnothing$
\State Set $S_v\gets\varnothing$ for every $v\in V$
\Statex \textit{Local search}
\Repeat
  \State $\mathrm{changed}\gets\mathrm{false}$
  \For{each $v\in V_{\mathrm{big}}$ with $h_v>0$, in a fixed order}
    \State $A\gets V_{\mathrm{small}}\setminus\bigcup_{u\ne v}S_u$
    \State $R'\gets$ Algorithm~\ref{alg:quota} on $(v,A)$
    \If{$\Phi(R')>\Phi(R_v)+\varepsilon_0h_v$}
      \State $R_v\gets R'$; $S_v\gets V(R')\setminus\{v\}$; $\mathrm{changed}\gets\mathrm{true}$
      \State \textbf{break} the for-loop
    \EndIf
  \EndFor
\Until{$\mathrm{changed}=\mathrm{false}$}
\State Keep these stored $R_v,S_v$ fixed; set $p'_v\gets p_v-s(R_v)$ for every $v$
\Statex \textit{Set replacements}
\ForAll{$S\subseteq V\setminus V_{\mathrm{small}}$ with $2\le|S|\le4$ and $d(S)\le1$}
  \State Compute $R_S$ by enumerating visit orders, and $w(S)$ by \eqref{eq:packing-weight}
\EndFor
\State Apply Theorem~\ref{thm:setpacking} to the sets with $w(S)>0$, obtaining $\mathcal S$
\State $\mathcal R_{\mathrm{pack}}\gets\{R_S:S\in\mathcal S\}$
\State $\mathcal R\gets\mathcal R_{\mathrm{pack}}\cup
 \{R_v:R_v\ne\varnothing,\ v\notin V(\mathcal R_{\mathrm{pack}})\}$
\State $X\gets V\setminus V(\mathcal R)$; $\mathcal T\gets\mathcal R$
\If{$X\ne\varnothing$}
  \State Form the restricted tour of $R_{\mathrm{TSP}}$ on $X\cup\{r\}$
  \State Compute a minimum-cost feasible consecutive partition by \eqref{eq:dp}
  \State Backtrack and add the resulting routes to $\mathcal T$
\EndIf
\State \Return $\mathcal T$
\end{algorithmic}
\end{algorithm}

\subsection{Local search}\label{sec:local}
This stage is reached only when $n\ge2$ and $D>0$.
When testing a replacement for $R_v$, define the available set
$A=V_{\mathrm{small}}\setminus\bigcup_{u\in V\setminus\{v\}}S_u$
from the current local-search state. It contains $S_v$ and all unassigned
small clients.
We bound the number of updates using
$\sum_{v:h_v>0}\Phi(R_v)/h_v$, which increases by more than
$\varepsilon_0$ at each accepted update.
Figure~\ref{fig:local-search} shows a replacement of $R_v$.

\begin{figure}[!t]
\centering
% Trim the standalone caption; use the numbered caption below.
\includegraphics[trim=0 29bp 0 0,clip]{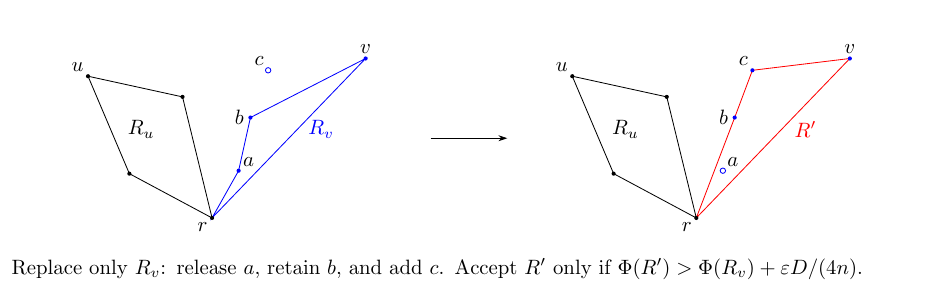}
\caption{Replace only $R_v$: release $a$, retain $b$, and add $c$.
Accept $R'$ only if $\Phi(R')>\Phi(R_v)+\varepsilon_0h_v$.}
\label{fig:local-search}
\end{figure}

The next lemma ensures that the stored routes remain feasible and
client-disjoint, and that the adjusted prices are nonnegative. It also
bounds the number of updates and gives the stopping condition used in
the cost analysis.
\begin{lemma}\label{lem:local}

Suppose $n\ge2$ and $D>0$.
The local-search stage of Algorithm~\ref{alg:global} satisfies the
following properties.
\begin{enumerate}
\item[(i)] At every stage, the nonempty routes $R_v$ are feasible and
have pairwise disjoint client sets. Each serves its big client $v$ and
a set $S_v$ of small clients. For every $v\in V$, we have
$0\le\Phi(R_v)\le s(R_v)\le p_v\le h_v$.
Every client $v$ with $d_v>1/2$ and $h_v>0$ has a nonempty
stored route throughout local search.
\item[(ii)] Local search terminates after fewer than
$n/\varepsilon_0$ accepted updates.
In the final complete scan, which accepts no update, let $R'_v$ be the
candidate returned by Algorithm~\ref{alg:quota} for each
$v\in V_{\mathrm{big}}$ with $h_v>0$.
Its available set is
$V_{\mathrm{small}}\setminus\bigcup_{u\ne v}S_u$,
computed from the final stored routes, and
\[
\Phi(R'_v)\le\Phi(R_v)+\varepsilon_0h_v.
\]
\end{enumerate}

\end{lemma}

The proof is given in Appendix~\ref{app:proof-local}.

\subsection{Set replacements}\label{subsec:packing-construction}

After local search, fix the stored routes $R_v$ and their small-client
sets $S_v$, and set $p'_v=p_v-s(R_v)$ for every $v\in V$.
Lemma~\ref{lem:local} and \eqref{eq:price} give
$0\le p'_v\le p_v\le2c(r,v)$.

Enumerate every set $S\subseteq V\setminus V_{\mathrm{small}}$ with
$2\le|S|\le4$ and $d(S)\le1$.
Compute a cheapest depot route $R_S$ visiting exactly $S$ by checking
its at most $24$ visit orders. Give $S$ the weight
\begin{equation}\label{eq:packing-weight}
w(S)=p(S)-c(R_S)-\sum_{v\in S}s(R_v)=p'(S)-c(R_S).
\end{equation}
Discard sets with nonpositive weight. Apply
Theorem~\ref{thm:setpacking} to the remaining family to obtain a
pairwise disjoint family $\mathcal S$, and let
$\mathcal R_{\mathrm{pack}}=\{R_S:S\in\mathcal S\}$.
Section~\ref{sec:packing} bounds its cost together with the adjusted
prices of the clients it does not serve.

To form $\mathcal R$, start with
$\mathcal R_{\mathrm{pack}}$ and retain every nonempty stored route
$R_v$ whose big client $v$ is not served by
$\mathcal R_{\mathrm{pack}}$.
The unserved set $X=V\setminus V(\mathcal R)$ includes the small
clients on every omitted stored route. Section~\ref{subsec:dp} serves
all clients in $X$.

Each packing route contains only non-small clients, so it can share
a client with $R_v$ only at $v$. The replacement rule omits every stored route
with such an intersection. The retained stored routes are already
client-disjoint, so $\mathcal R$ is feasible and client-disjoint.
Lemma~\ref{lem:replacement} proves the exact cost identity for this
replacement step using the adjusted prices $p'$.

\subsection{Tour partitioning}\label{subsec:dp}

If $X\ne\varnothing$, form the restricted
tour of $R_{\mathrm{TSP}}$ on $X\cup\{r\}$. Write its order
as $(r,v_1,\ldots,v_\ell,r)$, where $\ell=|X|$.

Let $q[i]$ be the minimum cost of partitioning $v_1,\ldots,v_i$
into consecutive feasible routes, with $q[0]=0$.
The recurrence is
\begin{equation}\label{eq:dp}
q[i]=
\min_{\substack{0\le j<i\\\sum_{t=j+1}^{i}d_{v_t}\le1}}
\left\{q[j]+c(r,v_{j+1})
+\sum_{t=j+1}^{i-1}c(v_t,v_{t+1})+c(v_i,r)\right\}.
\end{equation}
A minimizing predecessor for each $i$ determines the returned routes.
Prefix sums of demands and edge costs give $O(\ell^2)$ time and $O(\ell)$ space.
Figure~\ref{fig:tsp-partition} shows a consecutive partition and its routes.

\begin{figure}[!t]
\centering
% Trim the standalone caption; use the numbered caption below.
\includegraphics[trim=0 29bp 0 0,clip]{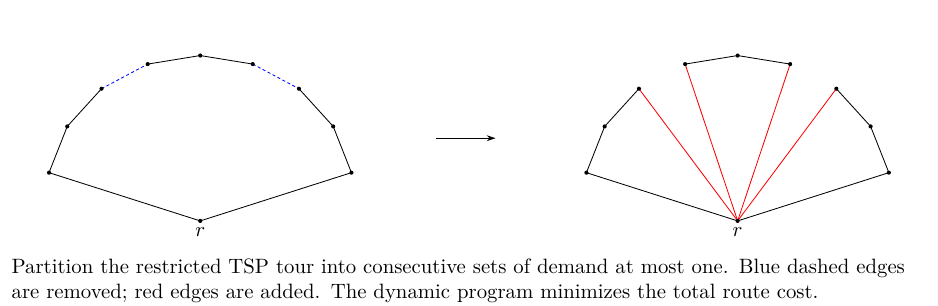}
\caption{Partition the restricted TSP tour into consecutive sets of demand
at most one. Blue dashed edges are removed; red edges are added.
The dynamic program minimizes the total route cost.}
\label{fig:tsp-partition}
\end{figure}

The following lemma combines singleton routes for clients with
$d_v>1/2$ and tour partitioning to obtain the general cost bound
defined by the client prices. Claim~\ref{clm:residual-dp} then shows
that for the unserved set $X$ in Algorithm~\ref{alg:global}, the
dynamic program~\eqref{eq:dp} achieves this bound directly without a
separate singleton-route step.
\begin{lemma}\label{lem:completion}
Let $\varnothing\ne X\subseteq V$, and let $R_{\mathrm{TSP}}$ be a given
TSP tour on $X\cup\{r\}$.
A feasible route family serving every client in $X$ exactly once,
with total cost at most $c(R_{\mathrm{TSP}})+p(X)$, can be computed in
$O(|X|^2)$ time and $O(|X|)$ additional storage.
\end{lemma}
\begin{proof}

Serve each client $v\in X$ with $d_v>1/2$ by $(r,v,r)$, at cost $p_v$,
and form the restricted tour on the remaining clients.
The triangle inequality ensures that the tour cost does not increase.
Apply Theorem~\ref{lem:tank} with $\delta=1/8$ to the remaining tour.
Each remaining client contributes
$\frac{16}{7}[d_v+(d_v-1/8)_+]c(r,v)=p_v$.
Adding the singleton costs gives the claimed bound.

The dynamic program~\eqref{eq:dp}, applied to the remaining client
order, computes a minimum-cost feasible consecutive partition, so its
cost is no greater than that of the
partition supplied by the theorem.
With prefix sums, it uses $O(|X|^2)$ time and
$O(|X|)$ additional storage.
\end{proof}

\begin{claim}\label{clm:residual-dp}

If $X\ne\varnothing$ after the set replacements in
Algorithm~\ref{alg:global}, the dynamic program~\eqref{eq:dp} applied
to the restricted tour on $X\cup\{r\}$ returns feasible routes of
total cost at most $c(R_{\mathrm{TSP}})+p(X)$.

\end{claim}

\begin{proof}

By Lemma~\ref{lem:local}, every client with $d_v>1/2$ and
$h_v>0$ has a nonempty stored route. The set replacements either
retain that route or serve $v$ on a selected packing route.
Thus every unserved client
with $d_v>1/2$ has $c(r,v)=0$, so the $\delta=1/8$ charge in
Theorem~\ref{lem:tank} equals $p_v$ for every $v\in X$.
The restricted tour costs at most $c(R_{\mathrm{TSP}})$.
The minimum-cost consecutive partition therefore attains the
claimed bound.
\end{proof}

\newpage

\section{Routes from quota trees}
\label{sec:quota}

\subsection{Route construction}\label{subsec:quota-construction}

Fix a client $v$ with $d_v>7/16$ and $h_v>0$, and an available
set $A\subseteq V_{\mathrm{small}}$.
We use the fixed cutoff $\tau_v=6c(r,v)/7$ and the filtered set
$A'=\{u\in A:c(r,u)\ge\tau_v\}$.
For every $X\subseteq A'$, summing the client prices gives
$p(X)\ge(16/7)\tau_v d(X)$.
In Lemma~\ref{lem:route-score}, we apply this inequality to the small
clients on the candidate route supplied by Lemma~\ref{lem:demand-recovery}
to obtain a lower bound on its score.
Rounding loses at most $\varepsilon_0/4$ total demand.
Each quota-tree call uses error $\varepsilon_0/4$.

Enumerate every set
$B\subseteq\{u\in A':d_u>\varepsilon_0/4\}$ with $d_v+d(B)\le1$.
Each such set has $|B|<4/\varepsilon_0$.
For each choice of $B$, the route must serve every client in
$B\cup\{v\}$. Call the clients $u\in A'$ with $d_u\le\varepsilon_0/4$
\emph{optional}. The greedy scan below determines which of these
clients to include.

For each optional client $u$, let
$m_u=\lfloor4nd_u/\varepsilon_0\rfloor$ and round $d_u$ down to
$\varepsilon_0 m_u/(4n)$. Create $m_u$ labelled copies of $u$.
Since $0\le d_u-\varepsilon_0 m_u/(4n)<\varepsilon_0/(4n)$ and there are
at most $n$ optional clients, the total decrease is at most
$\varepsilon_0/4$.
Let $m=\sum_{u\in A':\,d_u\le\varepsilon_0/4}m_u$.
Since each optional client has at most $n$ copies, $m\le n^2$.

Create $m+3$ copies of each client in $B\cup\{v\}$,
and add separate endpoint labels $r,v$.
All copies and endpoint labels are distinct
vertices of a complete graph. Edge lengths are inherited from the
original instance; copies of the same client have distance zero.
A client is represented in a tree if the tree contains a label
belonging to that client.

For each $k=0,\ldots,m$, require the tree to contain at least
$(|B|+1)(m+3)+2+k$ vertices, including $r$ and $v$.
A tree omitting all copies of some client in $B\cup\{v\}$
can contain at most $|B|(m+3)+m+2$ vertices, fewer than the quota.
Thus the tree contains a copy of every client in $B\cup\{v\}$.
These clients' copies and the two endpoint labels account for at most
$(|B|+1)(m+3)+2$ vertices, so at least $k$ vertices of the tree
must be copies of optional clients.
Figure~\ref{fig:quota-copy-graph} shows a tree in the copy graph.

\begin{figure}[!t]
\centering
% Trim the standalone caption; use the numbered caption below.
\includegraphics[trim=0 37.5bp 0 0,clip]{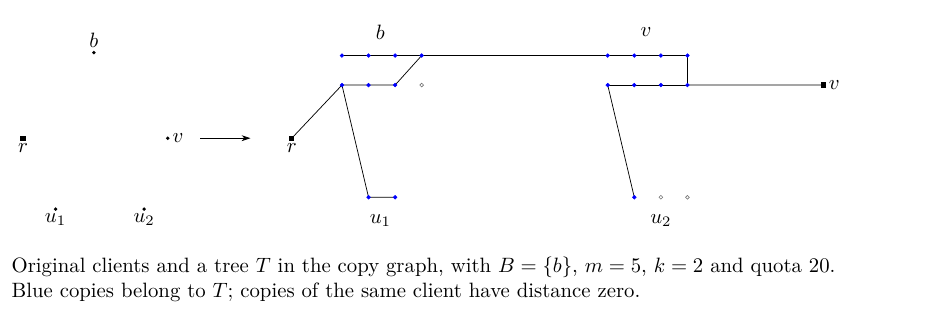}
\caption{Original clients and a tree $T$ in the copy graph, with $B=\{b\}$,
$m=5$, $k=2$ and quota $20$. Blue copies belong to $T$;
copies of the same client have distance zero.}
\label{fig:quota-copy-graph}
\end{figure}

For each returned tree $T$, select clients by the following
greedy rule. Initialize $S=B\cup\{v\}$. Scan, in a fixed order, each
optional client $u$ having at least one copy in $T$, and add $u$ to
$S$ if and only if $d(S)+d_u\le1$.
If $u$ is rejected, then $1-d(S)<d_u\le\varepsilon_0/4$.
Double every edge of $T$ outside its unique $r$--$v$ path, follow an
Euler trail of the resulting multigraph from $r$ to $v$, and append
$(v,r)$.
The $r$--$v$ path has cost at least $c(r,v)$, so the resulting
closed walk costs at most $2c(T)$.
Order the clients in $S\setminus\{v\}$ as
$u_1,\ldots,u_{|S|-1}$ by the first visit to a copy of each client along
this Euler trail. Set $R'=(r,u_1,\ldots,u_{|S|-1},v,r)$.
Figure~\ref{fig:tree-to-route} shows this construction.

\begin{algorithm}[!htb]
\caption{Quota Trees for Route Selection}
\label{alg:quota}
\begin{algorithmic}[1]
\Require A client $v$ with $d_v>7/16$ and $h_v>0$,
$A\subseteq V_{\mathrm{small}}$.
\Ensure A feasible route serving $v$ and a subset of $A$, with maximum
score $\Phi$ among the routes tested below.
\State $R_{\mathrm{best}}\gets(r,v,r)$.
\State $\tau_v\gets6c(r,v)/7$;
$A'\gets\{u\in A:c(r,u)\ge\tau_v\}$.
\State Set $m_u\gets\lfloor4nd_u/\varepsilon_0\rfloor$ for each
$u\in A'$ with $d_u\le\varepsilon_0/4$, and let $m$ be their sum.
\ForAll{$B\subseteq\{u\in A':d_u>\varepsilon_0/4\}$ with $d_v+d(B)\le1$}
  \State Build the complete graph with $m_u$ copies of each optional
  client, $m+3$ copies of each client in $B\cup\{v\}$, and endpoint
  labels $r,v$, using distances from the original instance.
  \For{$k=0,\ldots,m$}
    \State Apply Theorem~\ref{thm:quota-tree} with endpoints $r,v$,
    quota $(|B|+1)(m+3)+2+k$, and error $\varepsilon_0/4$ to obtain $T$.
    \State $S\gets B\cup\{v\}$.
    \ForAll{optional clients $u$ represented in $T$, in a fixed order}
      \If{$d(S)+d_u\le1$}
        \State $S\gets S\cup\{u\}$.
      \EndIf
    \EndFor
    \State Double every edge of $T$ outside its unique $r$--$v$ path,
    follow an Euler trail from $r$ to $v$, and append $(v,r)$.
    \State Order $S\setminus\{v\}$ as $u_1,\ldots,u_{|S|-1}$
    by the first visit to a copy of each client along this Euler trail.
    \State $R'\gets(r,u_1,\ldots,u_{|S|-1},v,r)$.
    \If{$\Phi(R')>\Phi(R_{\mathrm{best}})$}
      \State $R_{\mathrm{best}}\gets R'$.
    \EndIf
  \EndFor
\EndFor
\State \Return $R_{\mathrm{best}}$.
\end{algorithmic}
\end{algorithm}

Since $\varepsilon_0$ is fixed, one call to Algorithm~\ref{alg:quota}
takes $n^{O(1)}$ time. Appendix~\ref{sec:runtime} bounds the number of enumerated
sets $B$ and the number of calls to the quota-tree algorithm.

\subsection{Demand and cost bounds}

The next lemma compares a tested candidate with a reference $r$--$v$
path. It gives an upper bound on the candidate's cost and a lower bound
on the total demand of its small clients.
Lemma~\ref{lem:route-score} uses these bounds to obtain a score guarantee
for a tested candidate. Algorithm~\ref{alg:quota} returns a route with
at least that score.

\begin{lemma}
\label{lem:demand-recovery}
Fix an input $(v,A)$ to Algorithm~\ref{alg:quota}.
Let $\tau_v=6c(r,v)/7$ and $A'=\{u\in A:c(r,u)\ge\tau_v\}$.
For every set $X\subseteq A'$ satisfying $d_v+d(X)\le1$, and every
$r$--$v$ path $P$ visiting all clients in $X$, the algorithm tests a
feasible route $R'$ serving $v$ and a subset of $A'$ such that
\[
c(R')\le2(1+\varepsilon_0/4)c(P),
\qquad
d(V(R')\setminus\{v\})\ge d(X)-\varepsilon_0/4.
\]
Every route tested by the algorithm is feasible and serves $v$ and a
subset of $A$.
\end{lemma}

\begin{proof}
We choose the enumerated set and quota from $X$, so that $c(P)$ bounds
the cost of the returned tree up to the factor $1+\varepsilon_0/4$. To bound the selected demand, we
use two possibilities: the greedy scan accepts all optional clients
represented in the tree, or a rejected client shows that less than
$\varepsilon_0/4$ capacity remains.

Fix $X$ and $P$ satisfying the statement.
Consider the iteration
with $B=\{u\in X:d_u>\varepsilon_0/4\}$ and
$k=\sum_{u\in X\setminus B}m_u$.
This iteration is included because $d_v+d(B)\le1$ and $0\le k\le m$.

Retain the endpoints $r,v$ and the clients of $X$ in their order along
$P$, and connect consecutive retained vertices directly. The triangle
inequality ensures that the cost does not increase.
Replace each client in $X$ by all its copies, omitting clients with
zero copies, and insert all copies of $v$ immediately before the
endpoint label $v$. The resulting path in the copy graph has exactly
$(|B|+1)(m+3)+2+k$ vertices and cost at most $c(P)$.
Theorem~\ref{thm:quota-tree} therefore returns a tree $T$ with
$c(T)\le(1+\varepsilon_0/4)c(P)$.
Let $R'$ be the candidate constructed from $T$ in this iteration.

By the counting argument in Section~\ref{subsec:quota-construction}, $T$ contains a
copy of every client in $B\cup\{v\}$ and at least $k$ optional copies.
If an optional client $u$ has $t$ copies in $T$, then
$t\le m_u\le4nd_u/\varepsilon_0$, so $d_u\ge\varepsilon_0 t/(4n)$.
Summing over the optional clients represented in $T$, their total
demand is at least $\varepsilon_0 k/(4n)$.

If the greedy scan accepts all these clients, then
\[
\begin{aligned}
d(V(R')\setminus\{v\})
&\ge d(B)+\frac{\varepsilon_0 k}{4n}\\
&=d(B)+\sum_{u\in X\setminus B}
  \frac{\varepsilon_0}{4n}
  \left\lfloor\frac{4nd_u}{\varepsilon_0}\right\rfloor\\
&\ge d(X)-\frac{\varepsilon_0|X\setminus B|}{4n}
 \ge d(X)-\frac{\varepsilon_0}{4}.
\end{aligned}
\]
The last line uses $\lfloor z\rfloor\ge z-1$ and
$|X\setminus B|\le n$.

If an optional client $u$ is rejected, the remaining capacity at that
moment is smaller than $d_u\le\varepsilon_0/4$.
The remaining capacity can only decrease during the rest of the scan.
Hence the final selected small clients have total demand
$d(V(R')\setminus\{v\})>1-d_v-\varepsilon_0/4\ge d(X)-\varepsilon_0/4$,
where the last inequality follows from $d_v+d(X)\le1$.
This proves the demand bound.

To bound the route cost, let $P_T$ be the unique $r$--$v$ path in $T$.
Doubling all edges outside $P_T$ gives a connected multigraph whose
only odd-degree vertices are $r$ and $v$, so it has an Euler trail
from $r$ to $v$, as used in the construction. Append $(v,r)$ to this
trail. Retain the initial depot visit, the first visit to a copy of
each client in $V(R')\setminus\{v\}$, the terminal visit to $v$, and the
final depot visit. Connect consecutive retained vertices directly.
By the triangle inequality,
\[
c(R')\le2c(T)-c(P_T)+c(r,v)
\le2c(T)\le2(1+\varepsilon_0/4)c(P),
\]
using $c(P_T)\ge c(r,v)$.

Each iteration starts with the feasible set $B\cup\{v\}$ and adds a
client $u$ only when $d(S)+d_u\le1$.
Thus every constructed route is feasible and serves $v$ and a subset
of $A'$. Since $A'\subseteq A$, all tested routes serve $v$ and a subset
of $A$.
The initial route $(r,v,r)$ has the same properties.
\end{proof}

\begin{figure}[!t]
\centering
% Trim the standalone caption; use the numbered caption below.
\includegraphics[trim=0 29bp 0 0,clip]{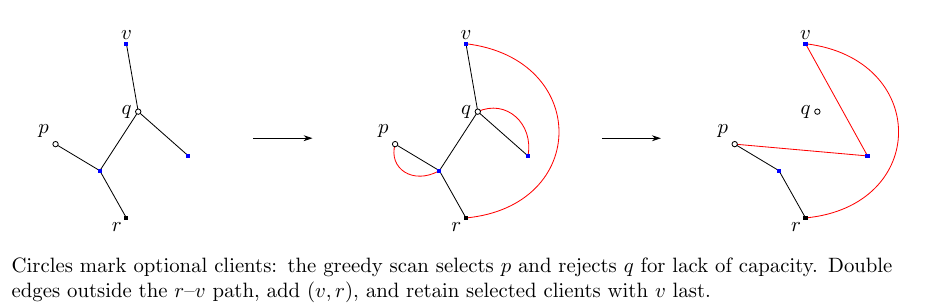}
\caption{Circles mark optional clients: the greedy scan selects $p$ and
rejects $q$ for lack of capacity. Double edges outside the $r$--$v$ path,
add $(v,r)$, and retain selected clients with $v$ last.}
\label{fig:tree-to-route}
\end{figure}

\section{Cost analysis}\label{sec:main-proof}
We bound the cost of the selected partial route family plus the prices
of the unserved clients, as required by \eqref{eq:selection-goal}.
Section~\ref{sec:packing} first compares set-packing savings with a
feasible reference solution and shows how the bound changes when prices
are reduced. Section~\ref{subsec:exceptional-routes} identifies the
optimal routes that need additional savings.
Section~\ref{subsec:candidate-scores} bounds the score of the route
returned by Algorithm~\ref{alg:quota}. Finally, Section~\ref{subsec:combining-savings} uses the final
local-search comparisons and the route-replacement identity to prove
\eqref{eq:selection-goal}.

\subsection{Set packing}\label{sec:packing}

We analyze the set-packing stage in
Section~\ref{subsec:packing-construction} for input prices $p'$.
For $j\in\{2,3,4\}$, put
$b_j=\frac9{10}\frac2{j+1}$, so
$b_2=3/5$, $b_3=9/20$, and $b_4=9/25$.
These are the coefficients supplied by
Theorem~\ref{thm:setpacking} for comparison sets of size $j$.

Fix a feasible reference route family $\mathcal R_0$ whose client
sets partition $V$. For a reference route $R$, let
$m=|V(R)\setminus V_{\mathrm{small}}|$.
If $m\ge2$, put $j=\min\{m,4\}$ and let $S_j(R)$ consist of the
$j$ non-small clients on $R$ with largest original demands, breaking
ties by a fixed order. These sets remain fixed when prices change.
Each is contained in a feasible route, and sets from distinct
reference routes are disjoint. For any nonnegative price vector
$p'$, define
\begin{equation}\label{eq:packing-contribution}
F(R;p')=
\begin{cases}
p'(V(R))-b_j\bigl(p'(S_j(R))-c(R)\bigr)_+,
   &m\ge2,\quad j=\min\{m,4\},\\
p'(V(R)),&m<2.
\end{cases}
\end{equation}
The subtracted term is the saving guaranteed by using $S_j(R)$ in
the comparison packing. The next lemma sums $F(R;p')$ over the
reference routes to bound the route cost plus the prices of unserved
clients, and shows how this sum changes when prices decrease.

\begin{lemma}\label{lem:packing-bound}
Let $\mathcal R_0$ be a feasible route family whose client sets
partition $V$, and let $F$ be defined by
\eqref{eq:packing-contribution}.
\begin{enumerate}
\item[(i)] For any nonnegative input prices $p'$, the route family
$\mathcal R_{\mathrm{pack}}$
returned by the set-packing stage satisfies
\[
c(\mathcal R_{\mathrm{pack}})
+p'\bigl(V\setminus V(\mathcal R_{\mathrm{pack}})\bigr)
\le\sum_{R\in\mathcal R_0}F(R;p').
\]
\item[(ii)] For any nonnegative price vectors $p''\le p'$ and every
$R\in\mathcal R_0$,
\[
F(R;p')-F(R;p'')\ge
\frac25\bigl(p'(V(R))-p''(V(R))\bigr).
\]
If $R$ has fewer than two non-small clients, the left-hand side equals
$p'(V(R))-p''(V(R))$.
\end{enumerate}
\end{lemma}

\begin{proof}
For part~(i), reference routes with positive subtracted terms provide
a disjoint comparison family for set packing. The cheapest route
through each comparison set costs at most its reference route.
For part~(ii), the subtracted term can offset at most a $b_j$
fraction of a price decrease, so $F$ decreases by at least $2/5$
of that decrease.

Consider a route $R\in\mathcal R_0$ with $m\ge2$ and
$p'(S_j(R))>c(R)$. Its set $S_j(R)$ is among those enumerated in
Section~\ref{subsec:packing-construction}. Retain its clients in
their visit order on $R$, together with the initial and final depot
visits, and connect consecutive retained vertices directly.
The triangle inequality bounds the resulting route cost by $c(R)$.
Since $R_{S_j(R)}$ is cheapest, the set weight satisfies
$w(S_j(R))=p'(S_j(R))-c(R_{S_j(R)})
\ge p'(S_j(R))-c(R)>0$.
Thus this set survives the removal of nonpositive weights.

All such sets are disjoint, so they form a comparison family for
Theorem~\ref{thm:setpacking}. The coefficient of a comparison set
of size $j$ is $b_j$. Thus the theorem gives the following lower
bound on $p'(V(\mathcal R_{\mathrm{pack}}))-c(\mathcal R_{\mathrm{pack}})$:
\[
\begin{aligned}
p'(V(\mathcal R_{\mathrm{pack}}))-c(\mathcal R_{\mathrm{pack}})
&=\sum_{S\in\mathcal S}w(S)\\
&\ge\sum_{R\in\mathcal R_0}
\bigl(p'(V(R))-F(R;p')\bigr).
\end{aligned}
\]
Subtracting this saving from $p'(V)$ proves part~(i), since the
reference routes partition $V$ and the subtracted terms in $F$ are
zero when $m<2$.

For part~(ii), fix $R$ and first suppose $m\ge2$.
The set $S_j(R)$ is unchanged because it was chosen by original
demands, rather than prices. For $x\in\mathbb R$ and $t\ge0$, the
positive-part function satisfies $0\le x_+-(x-t)_+\le t$.
Apply this with $x=p'(S_j(R))-c(R)$ and
$t=p'(S_j(R))-p''(S_j(R))$.
The decrease of the subtracted term is at most $b_j$ times the
price decrease on $S_j(R)$, which is at most the total decrease on
$V(R)$. Hence
\[
\begin{aligned}
F(R;p')-F(R;p'')
&\ge(1-b_j)\bigl(p'(V(R))-p''(V(R))\bigr)\\
&\ge\frac25\bigl(p'(V(R))-p''(V(R))\bigr),
\end{aligned}
\]
where the last inequality uses $b_j\le b_2=3/5$.
For $m<2$, there is no subtracted term, giving the stated equality.
\end{proof}

\subsection{Exceptional routes}\label{subsec:exceptional-routes}

We apply Lemma~\ref{lem:packing-bound} to the optimal route family
$\mathcal T^*$. The following lemma compares $F(R;p)$ with
$2\rho(R)+c(R)/2$. It also identifies the routes
on which local search must supply an additional saving.

\begin{lemma}\label{lem:exceptional-range}
Let $p$ be the prices in \eqref{eq:price}.
For every feasible route $R$, $F(R;p)\le2\rho(R)+c(R)/2$, unless
$R$ has exactly one non-small client $v$ with
$7/16<d_v<9/16$ and $c(r,v)>0$.
If $R$ has exactly one non-small client, then
$F(R;p)=p(V(R))$.
If also $p(V(R))>2\rho(R)+c(R)/2$, then
$c(R)<16\rho(R)/7$.
\end{lemma}

The proof is given in Appendix~\ref{app:proof-exceptional-range}.
For every route $R$, write
$\Delta(R)=p(V(R))-2\rho(R)-c(R)/2$.
Call a route of $\mathcal T^*$ exceptional if it has exactly one
non-small client and $\Delta(R)>0$.
Let $V_{\mathrm{exc}}$ be the set of these unique non-small clients,
and for each $v\in V_{\mathrm{exc}}$, let $R_v^*\in\mathcal T^*$ be
the exceptional route containing $v$.
Lemma~\ref{lem:exceptional-range} gives
$7/16<d_v<9/16$, $c(r,v)>0$, and
$c(R_v^*)<16\rho(R_v^*)/7$.

On an exceptional route, decreasing the price of its unique non-small
client decreases $F$ by the same amount. On every other route, $F$
decreases by at least $2/5$ of the total price decrease.
This gives the following bound at the adjusted prices.

\begin{lemma}\label{lem:packing-interface}
Suppose $0\le p'_v\le p_v$ for every $v\in V$.
Then the set-packing output $\mathcal R_{\mathrm{pack}}$ satisfies
\[
\begin{aligned}
&c(\mathcal R_{\mathrm{pack}})
 +p'\bigl(V\setminus V(\mathcal R_{\mathrm{pack}})\bigr)\\
&\quad\le\sum_{R\in\mathcal T^*}
       \left[2\rho(R)+\frac12c(R)\right]
 +\sum_{v\in V_{\mathrm{exc}}}
       \bigl[\Delta(R_v^*)-(p_v-p'_v)\bigr]
 -\frac25\sum_{v\notin V_{\mathrm{exc}}}(p_v-p'_v).
\end{aligned}
\]
\end{lemma}

The proof is given in Appendix~\ref{app:proof-packing-interface}.

\subsection{Candidate scores}\label{subsec:candidate-scores}
Lemma~\ref{lem:packing-interface} expresses the cost bound in terms
of excesses on exceptional routes and the savings from local search.
We use the fixed distance cutoff $6c(r,v)/7$ to convert recovered
demand into a lower bound on a candidate's small-client prices.
Applying demand recovery to the two depot-to-client paths of an
exceptional route then gives a score bound. The positive terms in
this bound exceed the errors in tree cost and recovered demand,
together with the local-search threshold.

The next lemma bounds the small-client prices by their total demand
and the demand above the fixed cutoff. Its proof compares the
contribution of each client separately.

\begin{lemma}\label{lem:distance-cutoff}
Let $v\in V$, let $H\ge h_v$, and let
$S\subseteq V_{\mathrm{small}}$ satisfy $h_u\le H$ for every $u\in S$.
Set $t=6h_v/7$ and $S_+=\{u\in S:h_u\ge t\}$.
Then
\begin{equation}\label{eq:distance-cutoff}
p(S)-\frac6{35}t\,d(S_+)
\le\left(\frac{8H}{7}-\frac{36h_v}{245}\right)d(S).
\end{equation}
\end{lemma}

\begin{proof}
We prove the inequality for each client's price per unit demand,
then multiply by its demand and sum.
For $u\in S_+$, the price formula and $h_u\le H$ give
$(8/7)h_u-(6/35)t\le8H/7-36h_v/245$.
For $u\in S\setminus S_+$, we have
$(8/7)h_u<(8/7)t=240h_v/245$.
Since $H\ge h_v$, this is at most
$244h_v/245\le8H/7-36h_v/245$.
Multiplying the corresponding inequality by $d_u$ for each client
and summing proves \eqref{eq:distance-cutoff}.
\end{proof}

The next lemma bounds $\Phi(R')$ for the route $R'$ returned by
Algorithm~\ref{alg:quota} in terms of the excess of an exceptional
route and the prices of its unavailable small clients.

\begin{lemma}\label{lem:route-score}
Let $R$ be an exceptional route with unique non-small client $v$.
Write $H=2\rho(R)$ and
$\Delta(R)=p(V(R))-H-c(R)/2>0$.
For an available set $A\subseteq V_{\mathrm{small}}$, let
$U=(V(R)\cap V_{\mathrm{small}})\setminus A$.
Algorithm~\ref{alg:quota} on $(v,A)$ returns a feasible route $R'$
that serves $v$ and a subset of $A$ and satisfies
\[
\Phi(R')\ge2\Delta(R)-\frac3{10}p(U)+3\varepsilon_0H.
\]
\end{lemma}

\begin{proof}
Apply Lemma~\ref{lem:demand-recovery} to the two $r$--$v$ paths of $R$.
Their costs sum to $c(R)$, so averaging the two score bounds gives
a cost term of $-(1+\varepsilon_0/4)c(R)$.
Lemma~\ref{lem:distance-cutoff} and the capacity constraint then bound
the small-client prices on $R$ in terms of the recovered demand.
The inequalities $7/16<d_v<9/16$ give positive terms larger than
the errors in tree cost and recovered demand.

Let $S=V(R)\cap V_{\mathrm{small}}$, $t=6h_v/7$, and
$S_+=\{u\in S:h_u\ge t\}$.
The algorithm uses $\tau_v=t/2$ and the filtered available set
$A'=\{u\in A:h_u\ge t\}$.
Every client $u\in A'$ satisfies $p_u\ge(8/7)t d_u$.
Lemma~\ref{lem:exceptional-range} gives $7/16<d_v<9/16$, $h_v>0$,
and $c(R)<8H/7$.

Split $R$ into two $r$--$v$ paths $P_1,P_2$, reversing the
second part of the route. Their costs sum to $c(R)$. Let $X_i$
be the small clients of $P_i$ that belong to $A'$.
These sets are disjoint and have union $V(R)\cap A'$.
For each $i$, feasibility gives $d_v+d(X_i)\le1$.
Lemma~\ref{lem:demand-recovery} supplies a tested route whose cost is
at most $2(1+\varepsilon_0/4)c(P_i)$ and whose small-client demand is
at least $d(X_i)-\varepsilon_0/4$.
Its score is therefore at least
$p_v+(24/35)t[d(X_i)-\varepsilon_0/4]
-2(1+\varepsilon_0/4)c(P_i)$.
The algorithm returns a route of maximum score among all tested routes.
Taking the average of the two bounds gives a demand term
$(12/35)t d(V(R)\cap A')$.

Since $S_+\setminus A'\subseteq U$ and
$p_u\ge(8/7)td_u$ for $u\in S_+$, we have
$(8/7)t[d(S_+)-d(V(R)\cap A')]\le p(U)$.
Consequently,
\[
\begin{aligned}
\Phi(R')\ge{}&p_v-c(R)+\frac{12}{35}t\,d(S_+)
-\frac3{10}p(U)\\
&-\frac{\varepsilon_0}{4}c(R)
-\frac{24}{35}t\,\frac{\varepsilon_0}{4}.
\end{aligned}
\]
The factor $1+\varepsilon_0/4$ contributes the loss
$(\varepsilon_0/4)c(R)$, and the demand loss contributes
$(24/35)t(\varepsilon_0/4)$. Since $(24/35)t=(144/245)h_v\le c(R)$,
their sum is at most $\varepsilon_0 c(R)/2$.

To apply the cutoff bound, we use feasibility and the price bound
$p_v\le2h_vd_v$. The coefficient $8H/7-36h_v/245$ is nonnegative,
and $d(S)\le1-d_v$. We use the following bound:
\[
\begin{aligned}
\left(\frac{8H}{7}-\frac{36h_v}{245}\right)(1-d_v)
&=H-h_vd_v-\frac{8d_v-1}{7}(H-h_v)
-\frac{h_v(1-d_v)}{245}\\
&\le H-\frac{p_v}{2}-\frac{8d_v-1}{7}(H-h_v)
-\frac{h_v(1-d_v)}{245}.
\end{aligned}
\]
Multiplying \eqref{eq:distance-cutoff} by two and using the preceding
bound gives a lower bound on the demand term. Substituting it into
the score estimate and using
$2p_v+2p(S)-2H-c(R)=2\Delta(R)$ gives
\[
\begin{aligned}
\frac{12}{35}t\,d(S_+)
&\ge2p(S)-2H+p_v
 +\frac{2(8d_v-1)}7(H-h_v)+\frac{2h_v(1-d_v)}{245},\\
\Phi(R')
&\ge2\Delta(R)-\frac3{10}p(U)
 +\frac{2h_v(1-d_v)}{245}
 +\frac{2(8d_v-1)}7(H-h_v)-\frac{\varepsilon_0}{2}c(R).
\end{aligned}
\]
Since $7/16<d_v<9/16$ and $h_v\le H$, the two positive terms are
at least $h_v/280+(5/7)(H-h_v)\ge H/280$.
Since $c(R)<8H/7$, the error is at most $(4/7)\varepsilon_0H$.
Our choice $\varepsilon_0\le10^{-3}$ gives
$1/280-(4/7)\varepsilon_0\ge3\varepsilon_0$, proving the stated bound.
\end{proof}

For an exceptional optimal route $R_v^*$, let
$U_v=(V(R_v^*)\cap V_{\mathrm{small}})\cap\bigcup_{u\ne v}S_u$
be its unavailable small clients in the final stored state.
The final scan tests $v$ by Lemma~\ref{lem:exceptional-range}.
Let $R'_v$ be the candidate returned for $v$ in this scan.
Its available set is $V_{\mathrm{small}}\setminus\bigcup_{u\ne v}S_u$,
so the unavailable set in Lemma~\ref{lem:route-score} is exactly $U_v$.
Applying that lemma to $R_v^*$ gives
$\Phi(R'_v)\ge2\Delta(R_v^*)-(3/10)p(U_v)+6\varepsilon_0\rho(R_v^*)$.
Lemma~\ref{lem:local}(ii) gives
$\Phi(R_v)\ge\Phi(R'_v)-\varepsilon_0h_v$.
Since $\varepsilon_0h_v\le2\varepsilon_0\rho(R_v^*)$, the final stored route satisfies
\begin{equation}\label{eq:stored-score}
\Phi(R_v)\ge2\Delta(R_v^*)-\frac3{10}p(U_v)+4\varepsilon_0\rho(R_v^*).
\end{equation}

\subsection{Combining savings}\label{subsec:combining-savings}
We now combine the bound on $\Phi(R')$ in Lemma~\ref{lem:route-score}
with the packing bound to control the cost of the selected routes plus the prices of
the unserved clients. We first express this cost using the adjusted
prices. We then use half of the bound on $\Phi(R_v)$ for each
$v\in V_{\mathrm{exc}}$ to offset $\Delta(R_v^*)$. The terms
$(2/5)p(S_v)$ in $s(R_v)=\Phi(R_v)+(2/5)p(S_v)$ cover the losses
from unavailable clients.

\begin{lemma}\label{lem:replacement}
Let $\mathcal R$ consist of the routes in $\mathcal R_{\mathrm{pack}}$
and the retained local-search routes. The adjusted prices are
$p'_v=p_v-s(R_v)$, where $R_v$ denotes the final stored route for $v$.
Then
\begin{equation}\label{eq:replacement}
c(\mathcal R)+p(V\setminus V(\mathcal R))
=c(\mathcal R_{\mathrm{pack}})
+p'(V\setminus V(\mathcal R_{\mathrm{pack}})).
\end{equation}
\end{lemma}

\begin{proof}
Since the selected routes have disjoint client sets, their total
saving is the sum of the savings of retained stored routes and routes
in $\mathcal R_{\mathrm{pack}}$.
Substituting $p_v=p'_v+s(R_v)$ for clients served by
$\mathcal R_{\mathrm{pack}}$ gives
\[
\begin{aligned}
p(V(\mathcal R))-c(\mathcal R)
&=\sum_{v\notin V(\mathcal R_{\mathrm{pack}})}s(R_v)
 +p(V(\mathcal R_{\mathrm{pack}}))-c(\mathcal R_{\mathrm{pack}})\\
&=\sum_{v\in V}s(R_v)
 +p'(V(\mathcal R_{\mathrm{pack}}))-c(\mathcal R_{\mathrm{pack}}).
\end{aligned}
\]
Subtracting this saving from $p(V)$ and using
$p(V)-\sum_v s(R_v)=p'(V)$ proves \eqref{eq:replacement}.
\end{proof}

The next lemma supplies the bound required by
\eqref{eq:selection-goal}.

\begin{lemma}\label{lem:selection}
Let $\mathcal R$ be the partial route family in
Algorithm~\ref{alg:global} after the set replacements. Then
\[
c(\mathcal R)+p(V\setminus V(\mathcal R))
\le\sum_{R\in\mathcal T^*}\left[2\rho(R)+\frac12c(R)\right].
\]
\end{lemma}

\begin{proof}
\eqref{eq:stored-score} bounds $2\Delta(R_v^*)$ in terms of
$\Phi(R_v)$ and $p(U_v)$.
We use half of this inequality. The terms $(2/5)p(S_v)$ in the
savings cover the loss from unavailable clients after summation,
giving the bound needed in Lemma~\ref{lem:packing-interface}.

Recall that $U_v$ consists of the unavailable small clients of
$R_v^*$ in the final stored state. The reference routes $R_v^*$ have
disjoint client sets.
Each unavailable client is therefore counted at most once, giving
$\sum_{v\in V_{\mathrm{exc}}}p(U_v)\le\sum_{u\in V}p(S_u)$.
Taking half of \eqref{eq:stored-score} and summing yields
\begin{equation}\label{eq:summed-scores}
\frac12\sum_{v\in V_{\mathrm{exc}}}\Phi(R_v)
\ge\sum_{v\in V_{\mathrm{exc}}}\Delta(R_v^*)
-\frac3{20}\sum_{u\in V}p(S_u)
+\varepsilon_0\sum_{v\in V_{\mathrm{exc}}}2\rho(R_v^*).
\end{equation}

The savings enter Lemma~\ref{lem:packing-interface} with coefficient
$1$ for $v\in V_{\mathrm{exc}}$ and $2/5$ otherwise.
For every $v$, the identity
$s(R_v)=\Phi(R_v)+(2/5)p(S_v)$ holds, and
$\Phi(R_v)\ge0$ by Lemma~\ref{lem:local}(i).
Thus every $p(S_v)$ has coefficient at least $(2/5)^2=4/25$, and
every $\Phi(R_v)$ with $v\in V_{\mathrm{exc}}$ has coefficient at least
$1/2$. Using \eqref{eq:summed-scores} and
$4/25-3/20=1/100$ gives
\begin{equation}\label{eq:weighted-savings}
\begin{aligned}
&\sum_{v\in V_{\mathrm{exc}}}s(R_v)
+\frac25\sum_{v\notin V_{\mathrm{exc}}}s(R_v)\\
&\quad\ge\frac12\sum_{v\in V_{\mathrm{exc}}}\Phi(R_v)
+\frac4{25}\sum_{v\in V}p(S_v)\\
&\quad\ge\sum_{v\in V_{\mathrm{exc}}}\Delta(R_v^*)
+\varepsilon_0\sum_{v\in V_{\mathrm{exc}}}2\rho(R_v^*)
+\frac1{100}\sum_{v\in V}p(S_v).
\end{aligned}
\end{equation}
Substitute this bound into Lemma~\ref{lem:packing-interface} with
$p'_v=p_v-s(R_v)$. Its price constraints follow from
Lemma~\ref{lem:local}(i). The $\Delta(R_v^*)$ terms cancel.
Applying \eqref{eq:replacement} gives
\[
\begin{aligned}
c(\mathcal R)+p(V\setminus V(\mathcal R))
&\le\sum_{R\in\mathcal T^*}\left[2\rho(R)+\frac12c(R)\right]\\
&\quad-\varepsilon_0\sum_{v\in V_{\mathrm{exc}}}2\rho(R_v^*)
-\frac1{100}\sum_{v\in V}p(S_v).
\end{aligned}
\]
Both subtracted terms are nonnegative, proving the claim.
\end{proof}

\section{Running time and approximation}
Section~\ref{sec:main-proof} bounds the cost of the selected routes
plus the prices of the unserved clients. We combine this cost bound
with the TSP tour and the completion guarantee to prove Theorem~\ref{thm:main}.

\subsection{Running time}
The running-time bound~\eqref{eq:runtime} is proved in Appendix~\ref{sec:runtime}.

\subsection{Proof of the main theorem}

\begin{proof}[Proof of Theorem~\ref{thm:main}]
We first handle the early-return cases, then check that the set
replacements and completion preserve feasibility.
For the cost bound, we add the TSP bound to the bound for selected
routes and remaining client prices. The terms $2\rho(R)$ cancel
for each optimal route.

If $n\le1$, the returned singleton routes are optimal.
If $D=0$, positivity of the demands implies $c(r,v)=0$ for every
client, so the returned routes have cost zero and are also optimal.
Henceforth assume $n\ge2$ and $D>0$.

The partial family $\mathcal R$ is feasible and client-disjoint by
Lemma~\ref{lem:local}(i) and the construction in
Section~\ref{subsec:packing-construction}.
The DP serves exactly the clients in $X=V\setminus V(\mathcal R)$
with feasible routes. Thus the returned family $\mathcal T$ serves
every client exactly once.

By Claim~\ref{clm:residual-dp}, completion adds cost at most
$c(R_{\mathrm{TSP}})+p(X)$ when $X\ne\varnothing$.
When $X=\varnothing$, no additional routes are needed, so this bound
also applies. Combining Lemma~\ref{lem:selection} with
Theorem~\ref{thm:full-tour} gives
\[
\begin{aligned}
c(\mathcal T)
&\le c(R_{\mathrm{TSP}})+c(\mathcal R)+p(X)\\
&\le\sum_{R\in\mathcal T^*}
\left[2c(R)-2\rho(R)+2\rho(R)+\frac12c(R)\right]
=\frac52\opt.
\end{aligned}
\]
\end{proof}

\section*{Acknowledgement}
Many thanks to Zachary Friggstad for his helpful discussions,
feedback, and suggestions on this work.

AI tools were used to assist with the research, the program-based checking
and the writing of this paper. The author is responsible for the
verification of all mathematical claims.

\bibliographystyle{siamplain}
{\raggedright
\bibliography{references}

\begin{thebibliography}{10}

\bibitem{AltinkemerGavish87}
{\sc K.~Altinkemer and B.~Gavish}, {\em Heuristics for unequal weight delivery
  problems with a fixed error guarantee}, Operations Research Letters, 6
  (1987), pp.~149--158, \url{https://doi.org/10.1016/0167-6377(87)90012-5}.

\bibitem{BeckerKS19}
{\sc A.~Becker, P.~N. Klein, and A.~Schild}, {\em A {PTAS} for bounded-capacity
  vehicle routing in planar graphs}, in Algorithms and Data Structures (WADS
  2019), vol.~11646 of Lecture Notes in Computer Science, Springer, 2019,
  pp.~99--111, \url{https://doi.org/10.1007/978-3-030-24766-9_8}.

\bibitem{Berman00}
{\sc P.~Berman}, {\em A {$d/2$} approximation for maximum weight independent
  set in {$d$}-claw free graphs}, Nordic Journal of Computing, 7 (2000),
  pp.~178--184,
  \url{https://www.cs.helsinki.fi/njc/References/berman:178.html}.

\bibitem{BlauthTraubVygen23}
{\sc J.~Blauth, V.~Traub, and J.~Vygen}, {\em Improving the approximation ratio
  for capacitated vehicle routing}, Mathematical Programming, 197 (2023),
  pp.~451--497, \url{https://doi.org/10.1007/s10107-022-01841-4}.

\bibitem{ChaudhuriGRT03}
{\sc K.~Chaudhuri, B.~Godfrey, S.~Rao, and K.~Talwar}, {\em Paths, trees, and
  minimum latency tours}, in Proceedings of the 44th Annual IEEE Symposium on
  Foundations of Computer Science (FOCS), IEEE Computer Society, 2003,
  pp.~36--45, \url{https://doi.org/10.1109/SFCS.2003.1238179},
  \url{https://www.microsoft.com/en-us/research/wp-content/uploads/2003/01/pub08-mlt.pdf}.

\bibitem{ChekuriKP12}
{\sc C.~Chekuri, N.~Korula, and M.~P{\'a}l}, {\em Improved algorithms for
  orienteering and related problems}, ACM Transactions on Algorithms, 8 (2012),
  pp.~23:1--23:27, \url{https://doi.org/10.1145/2229163.2229167},
  \url{https://chekuri.cs.illinois.edu/papers/orienteering-journal.pdf}.

\bibitem{Christofides22}
{\sc N.~Christofides}, {\em Worst-case analysis of a new heuristic for the
  travelling salesman problem}, Operations Research Forum, 3 (2022),
  \url{https://doi.org/10.1007/s43069-021-00101-z}.
\newblock Article 20.

\bibitem{DantzigRamser59}
{\sc G.~B. Dantzig and J.~H. Ramser}, {\em The truck dispatching problem},
  Management Science, 6 (1959), pp.~80--91,
  \url{https://doi.org/10.1287/mnsc.6.1.80}.

\bibitem{DasMathieu15}
{\sc A.~Das and C.~Mathieu}, {\em A quasipolynomial time approximation scheme
  for euclidean capacitated vehicle routing}, Algorithmica, 73 (2015),
  pp.~115--142, \url{https://doi.org/10.1007/s00453-014-9906-4}.

\bibitem{EbertWeismantel26}
{\sc D.~Ebert and L.~Weismantel}, {\em Improved approximation for unsplittable
  {CVRP} via a greedy approach}.
\newblock arXiv:2609.16910, 2026, \url{https://arxiv.org/abs/2609.16910}.

\bibitem{EiseltGL95}
{\sc H.~A. Eiselt, M.~Gendreau, and G.~Laporte}, {\em Arc routing problems,
  part {II}: The rural postman problem}, Operations Research, 43 (1995),
  pp.~399--414, \url{https://doi.org/10.1287/opre.43.3.399}.

\bibitem{FriggstadGrandoniMousavi26}
{\sc Z.~Friggstad, F.~Grandoni, and R.~Mousavi}, {\em Breaching the
  2-approximation barrier for euclidean capacitated vehicle routing}, in
  Proceedings of the 2026 Annual ACM-SIAM Symposium on Discrete Algorithms
  (SODA), Society for Industrial and Applied Mathematics, 2026, pp.~2644--2668,
  \url{https://doi.org/10.1137/1.9781611978971.97}.

\bibitem{FriggstadMomke26}
{\sc Z.~Friggstad and T.~M{\"o}mke}, {\em Approximating multiple-depot
  capacitated vehicle routing via {LP} rounding}, in Integer Programming and
  Combinatorial Optimization (IPCO 2026), vol.~16588 of Lecture Notes in
  Computer Science, Springer, 2026, pp.~112--127,
  \url{https://doi.org/10.1007/978-3-032-28691-8_8}.

\bibitem{FriggstadMRS25}
{\sc Z.~Friggstad, R.~Mousavi, M.~Rahgoshay, and M.~R. Salavatipour}, {\em
  Improved approximations for a {CVRP} with unsplittable demands}, Mathematics
  of Operations Research,  (2025),
  \url{https://doi.org/10.1287/moor.2022.0097}.
\newblock Articles in Advance.

\bibitem{GoldenWong81}
{\sc B.~L. Golden and R.~T. Wong}, {\em Capacitated arc routing problems},
  Networks, 11 (1981), pp.~305--315,
  \url{https://doi.org/10.1002/net.3230110308}.

\bibitem{GrandoniMathieuZhou23}
{\sc F.~Grandoni, C.~Mathieu, and H.~Zhou}, {\em Unsplittable euclidean
  capacitated vehicle routing: A {$(2+\varepsilon)$}-approximation algorithm},
  in 14th Innovations in Theoretical Computer Science Conference (ITCS 2023),
  vol.~251 of Leibniz International Proceedings in Informatics (LIPIcs),
  Schloss Dagstuhl -- Leibniz-Zentrum f{\"u}r Informatik, 2023,
  pp.~63:1--63:13, \url{https://doi.org/10.4230/LIPIcs.ITCS.2023.63}.

\bibitem{GurvitsKL24}
{\sc L.~Gurvits, N.~Klein, and J.~Leake}, {\em From trees to polynomials and
  back again: New capacity bounds with applications to {TSP}}, in 51st
  International Colloquium on Automata, Languages, and Programming (ICALP
  2024), vol.~297 of Leibniz International Proceedings in Informatics (LIPIcs),
  Schloss Dagstuhl -- Leibniz-Zentrum f{\"u}r Informatik, 2024,
  pp.~79:1--79:20, \url{https://doi.org/10.4230/LIPIcs.ICALP.2024.79}.

\bibitem{HaimovichRinnooyKan85}
{\sc M.~Haimovich and A.~H.~G. {Rinnooy Kan}}, {\em Bounds and heuristics for
  capacitated routing problems}, Mathematics of Operations Research, 10 (1985),
  pp.~527--542, \url{https://doi.org/10.1287/moor.10.4.527}.

\bibitem{Jansen93}
{\sc K.~Jansen}, {\em Bounds for the general capacitated routing problem},
  Networks, 23 (1993), pp.~165--173,
  \url{https://doi.org/10.1002/net.3230230304}.

\bibitem{JayaprakashS23}
{\sc A.~Jayaprakash and M.~R. Salavatipour}, {\em Approximation schemes for
  capacitated vehicle routing on graphs of bounded treewidth, bounded doubling,
  or highway dimension}, ACM Transactions on Algorithms, 19 (2023),
  pp.~20:1--20:36, \url{https://doi.org/10.1145/3582500}.

\bibitem{KarlinKO23}
{\sc A.~R. Karlin, N.~Klein, and S.~{Oveis Gharan}}, {\em A deterministic
  better-than-{$3/2$} approximation algorithm for metric {TSP}}, in Integer
  Programming and Combinatorial Optimization (IPCO 2023), vol.~13904 of Lecture
  Notes in Computer Science, Springer, 2023, pp.~261--274,
  \url{https://doi.org/10.1007/978-3-031-32726-1_19}.

\bibitem{KarlinKO24}
{\sc A.~R. Karlin, N.~Klein, and S.~{Oveis Gharan}}, {\em A (slightly) improved
  approximation algorithm for metric {TSP}}, Operations Research, 72 (2024),
  pp.~2543--2594, \url{https://doi.org/10.1287/opre.2022.2338},
  \url{https://doi.org/10.1287/opre.2022.2338}.

\bibitem{MathieuZhou23PTAS}
{\sc C.~Mathieu and H.~Zhou}, {\em A {PTAS} for capacitated vehicle routing on
  trees}, ACM Transactions on Algorithms, 19 (2023), pp.~17:1--17:28,
  \url{https://doi.org/10.1145/3575799}.

\bibitem{MathieuZhou23}
{\sc C.~Mathieu and H.~Zhou}, {\em A tight {$(1.5+\varepsilon)$}-approximation
  for unsplittable capacitated vehicle routing on trees}, Mathematical
  Programming, 212 (2025), pp.~115--146,
  \url{https://doi.org/10.1007/s10107-024-02094-z}.

\bibitem{MomkeZ23}
{\sc T.~M{\"o}mke and H.~Zhou}, {\em Capacitated vehicle routing in graphic
  metrics}, in Symposium on Simplicity in Algorithms (SOSA 2023), Society for
  Industrial and Applied Mathematics, 2023, pp.~114--123,
  \url{https://doi.org/10.1137/1.9781611977585.ch11}.

\bibitem{Sitters25}
{\sc R.~Sitters}, {\em An {$n^{O(\log\log n)}$} time approximation scheme for
  capacitated {VRP} in the {Euclidean} plane}.
\newblock arXiv:2507.15549v2, 2025, \url{https://arxiv.org/abs/2507.15549v2}.
\newblock Revised July 2026.

\bibitem{vanBevernHNS14}
{\sc R.~van Bevern, S.~Hartung, A.~Nichterlein, and M.~Sorge}, {\em
  Constant-factor approximations for {Capacitated Arc Routing} without triangle
  inequality}, Operations Research Letters, 42 (2014), pp.~290--292,
  \url{https://doi.org/10.1016/j.orl.2014.05.002}.

\bibitem{Wohlk08}
{\sc S.~W{\o}hlk}, {\em An approximation algorithm for the capacitated arc
  routing problem}, The Open Operational Research Journal, 2 (2008), pp.~8--12,
  \url{https://doi.org/10.2174/1874243200802010008}.

\bibitem{ZhaoXiaoCARP25}
{\sc J.~Zhao and M.~Xiao}, {\em An improved approximation algorithm for the
  capacitated arc routing problem}.
\newblock arXiv:2508.05471v1, 2025, \url{https://arxiv.org/abs/2508.05471v1}.

\bibitem{ZhaoXiaoMultidepot25}
{\sc J.~Zhao and M.~Xiao}, {\em Multidepot capacitated vehicle routing with
  improved approximation guarantees}, Theoretical Computer Science, 1043
  (2025), \url{https://doi.org/10.1016/j.tcs.2025.115265}.
\newblock Article 115265.

\bibitem{ZhaoXiao26}
{\sc J.~Zhao and M.~Xiao}, {\em Approximating unsplittable capacitated vehicle
  routing: Analysis via factor-revealing {LPs}}.
\newblock arXiv:2601.21660v2, 2026, \url{https://arxiv.org/abs/2601.21660v2}.

\bibitem{ZhaoXiaoFixed26}
{\sc J.~Zhao and M.~Xiao}, {\em Improved approximations for the capacitated
  vehicle routing problem with fixed capacity}, Information and Computation,
  308 (2026), \url{https://doi.org/10.1016/j.ic.2026.105405}.
\newblock Article 105405.

\bibitem{ZhaoXiao26Stochastic}
{\sc J.~Zhao and M.~Xiao}, {\em Tight approximations for fixed-order vehicle
  routing with stochastic demands}.
\newblock arXiv:2610.16802v1, 2026, \url{https://arxiv.org/abs/2610.16802v1}.

\end{thebibliography}
}
\appendix

\section{Omitted proofs}\label{app:omitted-proofs}

\subsection{Proof of Theorem~\ref{thm:setpacking}}\label{app:proof-setpacking}
\begin{proof}
We first analyze a packing at which no squared-weight exchange
improves the potential. We assign each comparison set to an
intersecting selected set of largest rounded weight $z$. An exchange of at most
four comparison sets bounds the total positive charge assigned to
one selected set. Summing these charges gives a bound with the
selected sets' actual sizes. We then undo the rounding, comparing
with a packing of maximum $q$-weight to control the total loss.

Suppose $N>0$; the empty-family case satisfies the theorem directly.
The integer potential $\sum_{T\in\mathcal S}z(T)^2$ strictly increases
after every accepted exchange. It is at most $N(10N)^2=100N^3$, so
there are at most $100N^3$ accepted exchanges. There are $O(N^4)$
incoming families to test per scan. Their pairwise disjointness takes
constant time to check because they contain at most four sets of
size at most four. Inspecting at most $N$ selected sets identifies those intersecting
the incoming family and determines the change in potential.
The total runtime is therefore $O(N^8)$.

Let $\mathcal S$ be the packing at termination, and let $\mathcal O$
be any packing of sets $S\in\mathcal F$ with $z(S)>0$. For each
$S\in\mathcal O$, let
$N(S)=\{T\in\mathcal S:T\cap S\ne\varnothing\}$.
This set is nonempty: otherwise adding $S$ increases the potential.
Choose $t(S)\in N(S)$ with largest $z(t(S))$, breaking ties by a fixed
order, and assign to it the charge
$\operatorname{ch}(S)=z(S)-\frac12\sum_{T\in N(S)}z(T)$.
For $T\in\mathcal S$, let
$\mathcal C_T=\{S\in\mathcal O:t(S)=T,
\operatorname{ch}(S)>0\}$.
The sets in $\mathcal C_T$ are disjoint and each meets $T$.
Consequently, $|\mathcal C_T|\le|T|\le4$, so their exchange is among
those tested by local search.

Fix $T$ with $\mathcal C_T\ne\varnothing$.
For $S\in\mathcal C_T$ and $U\in N(S)\setminus\{T\}$, the choice of
$t(S)$ gives $z(U)\le z(T)$, and hence $z(U)^2\le z(T)z(U)$.
Using also $(z(S)-z(T))^2\ge0$, we obtain
\[
\begin{aligned}
z(S)^2-\sum_{U\in N(S)\setminus\{T\}}z(U)^2
&\ge z(S)^2-z(T)\sum_{U\in N(S)\setminus\{T\}}z(U)\\
&\ge 2z(T)z(S)-z(T)^2
       -z(T)\sum_{U\in N(S)\setminus\{T\}}z(U)\\
&=2z(T)\operatorname{ch}(S).
\end{aligned}
\]
When these inequalities are summed over $\mathcal C_T$, a term
$z(U)^2$ with $U\ne T$ may be subtracted several times. Subtracting
each such term only once increases the left-hand side. Since the exchange with incoming
family $\mathcal C_T$ does not improve the potential, it follows that
\[
\begin{aligned}
2z(T)\sum_{S\in\mathcal C_T}\operatorname{ch}(S)
&\le\sum_{S\in\mathcal C_T}z(S)^2
   -\sum_{U\in N_{\mathcal S}(\mathcal C_T)\setminus\{T\}}z(U)^2\\
&\le z(T)^2.
\end{aligned}
\]
Thus the positive charges assigned to $T$ sum to at most $z(T)/2$.
If $T\in\mathcal O\cap\mathcal S$, then $\mathcal C_T=\{T\}$ and
$\operatorname{ch}(T)=z(T)/2$.

To sum the charges, note that each selected set $T$ intersects at
most $|T|$ members of the disjoint family $\mathcal O$.
Discarding negative charges and counting these intersections gives
\begin{equation}\label{eq:setpacking-rounded}
\begin{aligned}
\sum_{S\in\mathcal O}z(S)
&=\sum_{S\in\mathcal O}\operatorname{ch}(S)
  +\frac12\sum_{T\in\mathcal S}
     |\{S\in\mathcal O:S\cap T\ne\varnothing\}|\,z(T)\\
&\le\frac12\sum_{T\in\mathcal S}z(T)
  +\frac12\sum_{T\in\mathcal S}|T|z(T)
 =\sum_{T\in\mathcal S}\frac{|T|+1}{2}z(T).
\end{aligned}
\end{equation}

It remains to account for rounding. Let $\mathcal O^*$ be a packing
of maximum total $q$-weight in $\mathcal F$. Apply
\eqref{eq:setpacking-rounded} to the members of $\mathcal O^*$ with
$z(S)>0$; the omitted members contribute zero to the rounded sum.
For every $S\in\mathcal O^*$, we have $0\le q(S)-az(S)<a$,
and $|\mathcal O^*|\le N$. Moreover, $q(\mathcal O^*)\ge Q$, since a
single set of maximum $q$-weight is a packing. Therefore
\[
\begin{aligned}
\sum_{T\in\mathcal S}w(T)
&=\sum_{T\in\mathcal S}\frac{|T|+1}{2}q(T)
 \ge a\sum_{T\in\mathcal S}\frac{|T|+1}{2}z(T)\\
&\ge a\sum_{S\in\mathcal O^*}z(S)
 \ge q(\mathcal O^*)-Na
 =q(\mathcal O^*)-\frac Q{10}
 \ge\frac9{10}q(\mathcal O^*).
\end{aligned}
\]
Every comparison packing $\mathcal O\subseteq\mathcal F$ satisfies
$q(\mathcal O)\le q(\mathcal O^*)$. Substituting
$q(S)=2w(S)/(|S|+1)$ proves \eqref{eq:setpacking-guarantee}.
\end{proof}

\subsection{Proof of Lemma~\ref{lem:local}}\label{app:proof-local}
\begin{proof}
The available set prevents a replacement route from sharing clients with
other stored routes. We then bound each route's saving by its big client's price.
The sum of $\Phi(R_v)/h_v$ over clients with $h_v>0$ is bounded
and increases by more than $\varepsilon_0$ at every accepted update.

Initially, each nonempty route is a singleton route $(r,v,r)$ with
$d_v>1/2$ and $h_v>0$. These routes are feasible and client-disjoint.
Since $p_v=h_v=2c(r,v)$, their savings and scores are zero.
Empty routes also have zero saving and score.

Consider a candidate returned for $v$.
By the guarantee of Algorithm~\ref{alg:quota}, it is feasible and serves
$v$ together with clients from the available set.
Its big client $v$ differs from the big clients of all other stored
routes, and its small clients lie outside $\bigcup_{u\ne v}S_u$.
Replacing $R_v$ therefore preserves feasibility and disjointness, and
the new route still serves $v$ and only small clients besides $v$.
An accepted replacement strictly increases
$\Phi(R_v)$, so all scores remain nonnegative.
Every replacement for $R_v$ still serves $v$, so every initially
nonempty route remains nonempty.

To prove $s(R_v)\le p_v$, we bound the price of the small clients on
$R_v$. For a nonempty route, feasibility and $d_v>7/16$ give
$d(S_v)\le1-d_v<9/16$. Since $h_u\le c(R_v)$ for every
$u\in S_v$, the price definition gives
$p(S_v)\le(8/7)d(S_v)c(R_v)<9c(R_v)/14$.
In particular, $p(S_v)\le c(R_v)$, and therefore
$s(R_v)=p_v+p(S_v)-c(R_v)\le p_v$.
Moreover, $s(R_v)=\Phi(R_v)+(2/5)p(S_v)\ge\Phi(R_v)\ge0$.
Empty routes satisfy these bounds by convention, and
$p_v\le h_v$ follows from \eqref{eq:price}. This proves (i).

For the number of updates, consider the potential
\[
\sum_{v\in V:\,h_v>0}\frac{\Phi(R_v)}{h_v}.
\]
It starts at zero and is at most $n$ by part~(i).
An accepted update changes only one stored route and increases its
score by more than $\varepsilon_0h_v$. Thus the potential increases by more
than $\varepsilon_0$, and fewer than $n/\varepsilon_0$ updates can be accepted.
Each scan makes at most $n$ calls to Algorithm~\ref{alg:quota} and either
accepts an update or terminates, so local search terminates after at most
$n/\varepsilon_0+1$ scans.

The final scan accepts no update. All stored routes therefore remain
fixed throughout that scan, and every available set is computed from
the same final state. Each returned candidate fails the strict
acceptance test, which gives the inequality in (ii).
\end{proof}

\subsection{Proof of Lemma~\ref{lem:exceptional-range}}\label{app:proof-exceptional-range}
\begin{proof}
For routes with at least two non-small clients, increase prices to
the upper bounds from \eqref{eq:price} and $h_u\le2\rho(R)$.
For $c(R)>0$, we bound $(F(R;p)-2\rho(R))/c(R)$ by a
piecewise-linear function of one variable. With a unique non-small
client $v$, the price formula gives $7/16<d_v<9/16$ and
$c(R)<16\rho(R)/7$ whenever $F(R;p)>2\rho(R)+c(R)/2$.

If $c(R)=0$, every client on $R$ has distance zero from the depot, so
all prices and the claimed bounds vanish. Otherwise put $L=c(R)>0$,
$H=2\rho(R)$, and $m=|V(R)\setminus V_{\mathrm{small}}|$.
Every client $u$ on $R$ satisfies $2c(r,u)\le H\le L$.
If $H=0$, all prices again vanish. We may therefore assume $H>0$.
If $m=0$, all clients are small and feasibility gives
$F(R;p)=p(V(R))\le8H/7$. Consequently,
$F(R;p)-H\le H/7\le L/7\le L/2$.

Suppose $m\ge2$. Since each non-small demand is greater than $1/8$,
feasibility gives $2\le m\le7$.
For each non-small client $u$ on $R$, put
$\widetilde d_u=\min\{d_u,1/2\}$.
There is at most one demand above $1/2$, and all other demands are
then below $1/2$. Thus truncation preserves the original non-small
demand order. It also preserves the upper bound
$\sum_{u\in V(R)\setminus V_{\mathrm{small}}}\widetilde d_u
+d(V(R)\cap V_{\mathrm{small}})\le1$.

Leave prices outside $V(R)$ unchanged. On $R$, raise each non-small
price to $\bar p_u=H(16\widetilde d_u-1)/7$ and each small price to
$\bar p_u=(8/7)H d_u$.
The price formula and $2c(r,u)\le H$ give $\bar p_u\ge p_u$
for every $u\in V(R)$.
Lemma~\ref{lem:packing-bound}(ii) gives $F(R;p)\le F(R;\bar p)$.
The raised non-small prices have the original demand order, so
$S_j(R)$, for $j=\min\{m,4\}$, consists of the $j$ non-small
clients with largest raised prices.

Put $A_m=(16-m)/(7m)$ and
$z=\bar p(V(R)\setminus V_{\mathrm{small}})/(mL)$.
The total truncated demand of the non-small clients is
$(7mLz/H+m)/16$. Feasibility therefore gives
$d(V(R)\cap V_{\mathrm{small}})\le1-(7mLz/H+m)/16$.
Substituting this bound in the price formulas gives
$0\le z\le(H/L)A_m\le A_m$ and
$\bar p(V(R))/L\le mz/2+(H/L)mA_m/2$.
The largest $j$ among $m$ nonnegative numbers have at least $j/m$
of their total, so $\bar p(S_j(R))\ge jzL$.
Using this lower bound in the subtracted term of $F$ gives
\begin{equation}\label{eq:packing-profile}
\begin{aligned}
\frac{F(R;p)-H}{L}
&\le\frac HL\left(\frac{mA_m}{2}-1\right)
     +\frac{mz}{2}-b_j(jz-1)_+\\
&\le\left(m-\frac1{A_m}\right)z-b_j(jz-1)_+,
\qquad 0\le z\le A_m.
\end{aligned}
\end{equation}
For the second inequality, $mA_m/2-1=(2-m)/14\le0$ and
$H/L\ge z/A_m$. The expression on the last line is piecewise
linear, so its maximum occurs at $0$, at $A_m$, or at $1/j$ if
that point lies in the interval. Evaluating these points gives
\begin{center}
\begin{tabular}{c|c|c|c|c}
$m$ & $A_m$ & interior breakpoint & maximizing $z$ & maximum\\ \hline
$2$ & $1$ & $1/2$ & $1/2$ & $1/2$\\
$3$ & $13/21$ & $1/3$ & $13/21$ & $33/70$\\
$4$ & $3/7$ & $1/4$ & $3/7$ & $16/35$\\
$5$ & $11/35$ & $1/4$ & $11/35$ & $419/875$\\
$6$ & $5/21$ & $\varnothing$ & $5/21$ & $3/7$\\
$7$ & $9/49$ & $\varnothing$ & $9/49$ & $2/7$
\end{tabular}
\end{center}
Every maximum is at most $1/2$, proving the required bound for
$m\ge2$.

It remains to consider $m=1$. Let $v$ be the unique non-small client
and write $d=d_v$.
There is no subtracted term in $F$, so $F(R;p)=p(V(R))$.
The total price of the small clients is at most $(8/7)H(1-d)$.
For $d\le1/2$, the price $p_v$ is at most $H(16d-1)/7$;
for $d\ge1/2$, it is at most $H$.
Combining each bound with the bound on the small-client prices yields
\[
p(V(R))-H\le\frac87H\min\{d,1-d\}\le\frac47H.
\]
If $p(V(R))>H+L/2$, then $H\le L$ forces
$d>7/16$ and $1-d>7/16$, giving $7/16<d<9/16$.
The same inequality gives $L/2<4H/7$, and hence
$c(R)=L<8H/7=16\rho(R)/7$.
Finally, if $c(r,v)=0$, then $p_v=0$.
Since $d>1/8$, the small-client prices sum to less than
$(8/7)H(1-1/8)=H$, which rules out an excess.
Thus $c(r,v)>0$ whenever $p(V(R))>H+L/2$.
\end{proof}

\subsection{Proof of Lemma~\ref{lem:packing-interface}}\label{app:proof-packing-interface}
\begin{proof}
Apply the packing bound to $\mathcal T^*$ and account for price
decreases separately on exceptional and other routes. On an exceptional
route, decreasing the price of its unique non-small client decreases
$F$ by the same amount. On the other routes, use the $2/5$
bound from Lemma~\ref{lem:packing-bound}(ii).

For $v\in V_{\mathrm{exc}}$, the route $R_v^*$ has exactly one
non-small client, so $F(R_v^*;p')=p'(V(R_v^*))$. Every price decrease
on this route decreases $F$ by the same amount. Using only $2/5$ of
the decrease at clients other than $v$ gives
\[
\begin{aligned}
F(R_v^*;p')\le{}&2\rho(R_v^*)+\frac12c(R_v^*)+\Delta(R_v^*)
-(p_v-p'_v)\\
&-\frac25\sum_{u\in V(R_v^*)\setminus\{v\}}(p_u-p'_u).
\end{aligned}
\]

For every nonexceptional reference route $R$,
$F(R;p)\le2\rho(R)+c(R)/2$.
For zero or at least two non-small clients this follows from
Lemma~\ref{lem:exceptional-range}; with one non-small client it
follows from the definition of exceptionality.
Lemma~\ref{lem:packing-bound}(ii) therefore gives
\[
F(R;p')\le2\rho(R)+\frac12c(R)
-\frac25\sum_{v\in V(R)}(p_v-p'_v).
\]
The client sets $V(R)$, $R\in\mathcal T^*$, partition $V$.
Summing these bounds retains
the full price decrease at every $v\in V_{\mathrm{exc}}$ and $2/5$
of the decrease at every other client. Applying
Lemma~\ref{lem:packing-bound}(i) proves the stated bound.
\end{proof}

\subsection{Running time}\label{sec:runtime}
The early-return cases $n\le1$ and $D=0$ take $O(n)$ time.
Assume $n\ge2$ and $D>0$.

Recall that $\varepsilon_0$ is a fixed rational constant, independent of the input.
Each client in an enumerated set $B$ has $d_u>\varepsilon_0/4$.
Together with $d_v+d(B)\le1$, this gives $|B|<4/\varepsilon_0$.
As in Algorithm~\ref{alg:quota}, let $m$ be the total number of optional
client copies, and let $N_{\mathrm{copy}}$ be the number of vertices
in the copy graph for the current $B$.
Every optional client has at most $n$ copies. Hence
\[
\begin{aligned}
\#\text{enumerated sets }B
&\le\sum_{\substack{0\le j\le n\\j<4/\varepsilon_0}}\binom nj=n^{O(1/\varepsilon_0)},\\
m&\le n^2,\qquad
N_{\mathrm{copy}}=m+(|B|+1)(m+3)+2
\le(n+2)(n^2+3)=O(n^3).
\end{aligned}
\]
Each call to Algorithm~\ref{alg:quota} uses one cutoff.
For each enumerated set $B$,
at most $n^2+1$ quota values are tested.
Lemma~\ref{lem:local}(ii) gives at most $n/\varepsilon_0+1$ scans,
with at most $n$ calls to Algorithm~\ref{alg:quota} per scan. The total number of
quota-tree calls is therefore at most
\[
n(n/\varepsilon_0+1)
\left(\sum_{\substack{0\le j\le n\\j<4/\varepsilon_0}}\binom nj\right)(n^2+1).
\]
At accuracy $\varepsilon_0/4$,
Theorem~\ref{thm:quota-tree} takes
$N_{\mathrm{copy}}^{O(1/\varepsilon_0)}=n^{O(1/\varepsilon_0)}$ time per call.
Constructing the copy graph and recovering a route take polynomial
time, so the entire local-search stage takes $n^{O(1/\varepsilon_0)}$ time.

There are $N=O(n^4)$ sets with positive weights $w(S)$.
For each set $S$, computing $R_S$ checks at most $4!=24$ visit orders,
and computing $w(S)$ uses at most four stored savings.
Theorem~\ref{thm:setpacking} has runtime $O(N^8)$.

There is one $O(n^7)$ call to Theorem~\ref{thm:full-tour}.
Forming the restriction of its tour to $X$ takes $O(n)$ time.
The final DP takes $O(n^2)$ time and $O(n)$ space.
This proves \eqref{eq:runtime}.

\section{Capacitated arc routing}\label{app:carp}

\subsection{Service tours}\label{carp:sec:tours}
\begin{definition}[\textsc{Capacitated Arc Routing}]\label{def:carp}
An instance of \textsc{Capacitated Arc Routing} (CARP) consists of:
\begin{itemize}[nosep]
    \item a complete undirected graph $G=(V,E)$ and a depot $r\in V$,
    \item nonnegative rational metric costs $c:E\to\mathbb R_{\ge0}$,
    \item a set of required edges $E^+\subseteq E$, where $n=|E^+|$,
    \item a rational demand $d_e\in(0,1]$ for each $e\in E^+$.
\end{itemize}
A route is a walk in $G$. A feasible solution is a collection of
routes $\mathcal T$ such that
\begin{itemize}[nosep]
    \item every route $R\in\mathcal T$ starts and ends at $r$,
    \item every required edge is assigned to exactly one route, which
    serves its demand in one designated traversal of the entire edge,
    \item $\sum_{e\in E(R)}d_e\le1$ for every $R\in\mathcal T$, where
    $E(R)$ is the set of required edges assigned to $R$.
\end{itemize}
Other traversals do not serve demand. The cost $c(R)$ counts all
traversals with multiplicity. The objective is to minimize
$c(\mathcal T)=\sum_{R\in\mathcal T}c(R)$.
Fix an optimal solution $\mathcal T^*$ and write $\opt=c(\mathcal T^*)$.
\end{definition}

The approximation-preserving reduction of
van Bevern et al.~\cite[Theorem~1]{vanBevernHNS14} reduces nonnegative
costs to this metric case. Rational inputs can first be scaled to
integers.

For $S\subseteq E^+$, write $d(S)=\sum_{e\in S}d_e$.
Give each required edge its own colocated endpoint copies and retain only
these endpoints and $r$. Inherited distances preserve costs and make
the required edges a matching with no endpoint equal to $r$;
distinct vertices may have zero distance.
For $e=ab\in E^+$, write $c_e=c(e)$ and $h_e=c(r,a)+c_e+c(r,b)$, the cost
of the singleton route $(r,a,b,r)$. For $E(R)\ne\varnothing$, put
$\rho(R)=\frac12\max_{e\in E(R)}h_e$, and set $\rho(\varnothing)=0$.
For $e\in E(R)$, replacing the walks before and after its designated
traversal by direct depot connections gives $h_e\le c(R)$. Hence $2\rho(R)\le c(R)$, and summing with
$d(E(R))\le1$ gives $D:=\sum_{e\in E^+}d_eh_e\le\opt$.

Apply the demand classes in Definition~\ref{def:classes} to required edges,
and write $E_{\mathrm{small}}$ for the small required edges. Define $p_e$
by \eqref{eq:price}, replacing $d_v$ by $d_e$ and $c(r,v)$ by
$h_e/2$. These prices satisfy $0\le p_e\le h_e$ and $p_e\le2d_eh_e$.
Use $p(S)=\sum_{e\in S}p_e$ and
$E(\mathcal R)=\bigcup_{R\in\mathcal R}E(R)$.
Define $s(R)=p(E(R))-c(R)$ and
$\Phi(R)=s(R)-(2/5)p(E(R)\cap E_{\mathrm{small}})$,
with both zero for the empty route.
Disjointness refers to the assigned subsets of $E^+$; routes may
traverse the same edges.

For $X\subseteq E^+$, a \emph{service tour} for $X$ is a closed
depot walk with one designated traversal of each edge in $X$, without
a capacity constraint.
To restrict a route or service tour to some served edges, retain their
designated traversals in order and direction. Connecting them and the depot at both
ends by shortest paths does not increase cost.
The construction of Zhao and Xiao~\cite[Section~6.4]{ZhaoXiao26Stochastic}
returns a service tour $R_{\mathrm{RPP}}$ for $E^+$ satisfying
\begin{equation}\label{carp:eq:full-tour}
c(R_{\mathrm{RPP}})\le
\sum_{R\in\mathcal T^*}\bigl(2c(R)-2\rho(R)\bigr).
\end{equation}
On each reference route $R$, choose $e$ maximizing $h_e$.
Deleting its designated traversal gives two depot-rooted walks of total
cost $c(R)-c_e$. Retaining their designated traversals and connecting
them directly gives paths of no greater cost. Their Section~6.4 bounds
the output cost by the sum of
$2(c(R)-c_e)+(2c_e-h_e)=2c(R)-2\rho(R)$ over the reference routes.

The completion bound requires cuts before or after whole required edges.
\begin{lemma}\label{carp:lem:completion}
Given a service tour $R_X$ for $X\subseteq E^+$, a feasible route family
serving $X$ with total cost at most $c(R_X)+p(X)$ can be computed in
polynomial time.
\end{lemma}

\begin{proof}
We apply the shifted demand grid from Theorem~\ref{lem:tank}, cutting
before or after whole required edges. The shift determines the
probability of each cut, which bounds the expected added cost.

Serve required edges with $d_e>1/2$ separately at cost $h_e=p_e$.
The restriction of $R_X$ to the remaining required edges has cost
at most $c(R_X)$. If none remain, this proves the claim.
List the remaining edges as $e_i=u_iv_i$ in the order and directions
of their designated traversals in $R_X$, and connect consecutive
traversals and the depot at both ends by shortest paths.
Set $a_i=\sum_{j=1}^i d_{e_j}$, with $a_0=0$.
The interval $(a_{i-1},a_i]$ represents the demand of $e_i$.
For $\delta=1/8$, choose $\omega$ uniformly from $[0,1-\delta)$ and
place grid points $\omega+j(1-\delta)$, $j\in\mathbb Z$, on the
demand line. Define
\[
\begin{aligned}
B_i&=\{\omega\in[0,1-\delta):\exists j\in\mathbb Z,
          \ a_{i-1}<\omega+j(1-\delta)\le a_i\},\\
I_i&=\{\omega\in[0,1-\delta):\exists j\in\mathbb Z,
          \ a_{i-1}<\omega+j(1-\delta)<a_i-\delta\}.
\end{aligned}
\]
Thus $B_i$ means that the demand interval contains a grid point;
$I_i$ means that this point lies more than $\delta$ before its right
endpoint. Since $d_{e_i}\le1/2<1-\delta$, each interval contains
at most one grid point. Reducing such an interval modulo $1-\delta$
gives a subset of $[0,1-\delta)$ of the same total length.
The shifts in $B_i$ therefore have total length $d_{e_i}$, and those
in $I_i$ have total length $(d_{e_i}-\delta)_+$.
Uniformity gives $\Pr(B_i)=d_{e_i}/(1-\delta)$ and
$\Pr(I_i)=(d_{e_i}-\delta)_+/(1-\delta)$.

On $B_i\setminus I_i$, cut after $e_i$; on $I_i$, cut before and
after it, isolating the required edge. These cuts divide the service
order into consecutive segments. Join each nonempty segment to the
depot at both ends, and let $\mathcal R_\omega$ be the resulting
route family.

To check capacity, let $b$ be a grid point in $(a_{i-1},a_i]$.
On $B_i\setminus I_i$, the cut is at cumulative demand $a_i$,
where $b\le a_i\le b+\delta$. On $I_i$, the two cuts are at
$a_{i-1}<b$ and $a_i>b+\delta$, with $e_i$ alone between them.
Thus the first cut associated with $b$ is at most $b+\delta$,
and the last is at least $b$. The next grid point is
$b+(1-\delta)$, so a segment between the two groups of cuts has
demand at most $[b+(1-\delta)+\delta]-b=1$.
Isolated required edges have demand at most $1/2$.
The first positive grid point is at most $1-\delta$ from zero,
so the initial segment also has demand at most one. The final segment
starts at or after the last grid point in the demand intervals and
ends before the next, so its demand is at most $1-\delta$. If no
demand interval contains a grid point, total demand is at most
$1-\delta$.

For the cost, a cut between consecutive served edges $e_i,e_{i+1}$
replaces the connection from $v_i$ to $u_{i+1}$ by two depot
connections. If charged to the cut after $e_i$, its added cost is
\[
c(v_i,r)+c(r,u_{i+1})-c(v_i,u_{i+1})
\le2c(r,v_i)\le h_{e_i}.
\]
The first inequality follows from the triangle inequality, and the
second uses $c(r,v_i)\le c(r,u_i)+c_{e_i}$.
For a cut before $e_i$, the corresponding bound is
$2c(r,u_i)\le h_{e_i}$.
Cuts before the first or after the last served edge add no cost.
If two rules mark the same connection, cut it only once; counting
both marks still gives an upper bound.

Every $B_i$ marks a cut after $e_i$, and $I_i\subseteq B_i$ marks
an additional cut before it. Let $1_{B_i}$ and $1_{I_i}$ denote
their indicator functions. Summing the cut costs and then taking
expectations gives
\[
\begin{aligned}
c(\mathcal R_\omega)
&\le c(R_X)+\sum_i h_{e_i}
       \bigl(\mathbf{1}_{B_i}(\omega)+\mathbf{1}_{I_i}(\omega)\bigr),\\
\mathbb E[c(\mathcal R_\omega)]
&\le c(R_X)+\sum_i h_{e_i}\bigl(\Pr(B_i)+\Pr(I_i)\bigr)\\
&=c(R_X)+\sum_i\frac{h_{e_i}}{1-\delta}
       \bigl[d_{e_i}+(d_{e_i}-\delta)_+\bigr]
 =c(R_X)+\sum_i p_{e_i}.
\end{aligned}
\]
Some offset therefore gives a feasible route family of cost at most
$c(R_X)+\sum_i p_{e_i}$.

Every offset produces a feasible partition into consecutive segments
of the given service order, so each such partition is among those
considered by \eqref{eq:dp}. Apply that recurrence, counting each
segment's edge lengths, shortest connections between consecutive
edges, and depot connections at both ends. The DP takes quadratic
time using prefix sums and returns a partition of no greater cost.
Adding the singleton routes for $d_e>1/2$, each of cost $p_e$,
gives the bound $c(R_X)+p(X)$.
\end{proof}

\subsection{Local routes}\label{carp:sec:local}
To serve whole edges, modify Algorithm~\ref{alg:quota} as follows.
For a big required edge $e_0=ab$ with $h_{e_0}>0$ and an available set
$A\subseteq E_{\mathrm{small}}$, set $\theta=\varepsilon_0/16$ and
$A'=\{e\in A:h_e\ge6h_{e_0}/7\}$.
Apply the construction at both endpoints $t\in\{a,b\}$.
\begin{enumerate}[nosep]
\item Enumerate $B\subseteq\{e\in A':d_e>\theta\}$ with
$d_{e_0}+d(B)\le1$. For each $e\in A'$ with $d_e\le\theta$, put
$m_e=\lfloor nd_e/\theta\rfloor$ internal labels at
$jc_e/(m_e+1)$, $j=1,\ldots,m_e$, measured from a fixed endpoint.
Let $m$ be their total number.
Put $M=\lceil(m+1)/\theta\rceil$ equally spaced internal labels
on each $e\in B$. Replace each labelled edge by the path through
its endpoints and these labels, in order. Internal vertices have no
other incident edges; labels on zero-length edges remain distinct.
Use the complete graph on the labels and $r,t$, with shortest-path
distances in the subdivided graph. For $k=0,\ldots,m$, use quota
$|B|M+2+k$ and accuracy $\theta$. Expand the returned tree into
shortest paths and take a spanning tree $T$ preserving all visited
vertices.

\item Start with a copy of $T$. Add the missing parts of the paths
replacing edges in $B$.
For $e\in A'$ with $d_e\le\theta$ and $c_e=0$, also add the entire
path if $T$ visits one of its labels.
For each $e\in A'$ with $d_e\le\theta$, let $\ell_e$ be the total
length of the edges of its subdivided path that lie in $T$.
On $\mathcal F=\{e\in A':d_e\le\theta,\ 0<\ell_e<c_e\}$, solve
a fractional knapsack with weights $c_e$, profits $d_e$, and budget
$W=\sum_{e\in\mathcal F}\ell_e$. Use decreasing $d_e/c_e$ order
and discard the one fractional item, if present.
For every $e\in\mathcal F$, remove its portions in $T$ and their
internal vertices, retaining original endpoints already present.
Insert the entire paths for the edges selected by the knapsack.
Denote the resulting subgraph by $K$, and let $c(K)$ be the sum
of its edge costs.

\item Initialize $S=B$. Scan the edges $e\in A'$ with $d_e\le\theta$
whose entire subdivided paths lie in $K$, adding $e$ to $S$ only if
$d_{e_0}+d(S)+d_e\le1$. Keep every rejected edge's path in $K$
for travel without serving its demand.

\item Replace each complete subdivided path in $K$ by its original
edge, keeping the name $K$. Double every edge of $K$ and remove
one copy of a simple $r$--$t$ path.
Follow an Euler trail from $r$ to $t$, traverse $e_0$, and return
from its other endpoint to $r$. Serve $S$ on complete edge traversals
and $e_0$ on the appended traversal.
\end{enumerate}
In Step~4, every edge of $K$ retains at least one copy, so the
multigraph is connected. Removing one copy of a simple $r$--$t$
path from the doubled graph makes $r,t$ its only odd-degree vertices,
giving the required Euler trail.
Use fixed orders for enumeration and tie-breaking. As in
Algorithm~\ref{alg:quota}, return a route $R'$ maximizing $\Phi(R')$
among the singleton route and all candidates from both endpoints.

The next lemma replaces Lemma~\ref{lem:demand-recovery} for CARP.
\begin{lemma}\label{carp:lem:recovery}
Fix $e_0,A',t$ as above. For every $X\subseteq A'$ with
$d_{e_0}+d(X)\le1$ and every $r$--$t$ walk $P$ with one designated
traversal of each edge in $X$, the construction generates a connected
subgraph $K$ containing
$r,t$ and all edges of a set $S\subseteq A'$ such that
\begin{equation}\label{carp:eq:recovery}
c(K)\le(1+\varepsilon_0/4)c(P),\qquad
d(S)\ge d(X)-\varepsilon_0/4,\qquad d_{e_0}+d(S)\le1.
\end{equation}
Every generated subgraph is connected and contains the whole edges
selected for service, whose total demand fits with $e_0$.
\end{lemma}

\begin{proof}
The label counts bound the missing length and demand loss.
The knapsack replaces segments ending at internal leaves by whole
paths without increasing cost.

Choose $B=\{e\in X:d_e>\theta\}$ and
$k=\sum_{e\in X\setminus B}m_e$. Retain the initial $r$, final $t$,
and one occurrence of each internal label on an edge in $X$.
Connecting them in their order gives an $r$--$t$ path meeting the quota
at cost at most $c(P)$, so Theorem~\ref{thm:quota-tree} gives
$c(T)\le(1+\theta)c(P)$ after expansion and taking a spanning tree.
There are at least $k$ labels on edges with $d_e\le\theta$, since
the labels on $B$ and the endpoints account for at most $|B|M+2$
vertices. For each $e\in B$, $T$ contains at least $M-m>0$ of its labels, since all
other vertices of the complete graph account for at most
$(|B|-1)M+m+2$.

Internal subdivision vertices have no other incident edges. If $T$
does not contain the entire path replacing an edge, its intersection
with that path consists of at most two segments, each joining an
original endpoint to an internal leaf. For $e\in A'$ with
$d_e\le\theta$, let $N_e$ count its visited labels.
A complete path, or these retained segments, contains at least as
many subdivision edges as internal labels. Each subdivision edge
has length $c_e/(m_e+1)$, so
$\ell_e/c_e\ge N_e/(m_e+1)$ when $c_e>0$.
For $e\in B$, at least $M-m$ labels are visited, so at most $m+1$
of the $M+1$ subdivision edges are missing. Their total length is
at most $(m+1)c_e/(M+1)\le\theta c_e$. Adding the missing parts therefore
costs at most $\theta\sum_{e\in B}c_e\le\theta c(P)$.

For every call, deleting the portions of edges in $\mathcal F$
preserves connectivity among $r,t$ and all original vertices of $T$:
paths between them cannot enter segments ending at internal leaves.
Each completed or inserted path meets the retained graph at an original
endpoint, so $K$ is connected.

The fractions $\ell_e/c_e$, $e\in\mathcal F$, form a feasible
knapsack solution of weight $W$. The optimum has at least their
profit. The selected whole paths cost at most $W$, the length removed,
so this replacement does not increase cost.

For $e\in\mathcal F$, uniform spacing gives
\[
d_e\frac{\ell_e}{c_e}
\ge\frac{\theta}{n}\frac{m_eN_e}{m_e+1}
\ge\frac{\theta}{n}(N_e-1).
\]
Each $e\in A'$ with $d_e\le\theta$ whose entire path is retained before
the knapsack step contributes at least $\theta N_e/n$ demand;
this includes zero-length paths added at zero cost.
Their total demand plus $\sum_{e\in\mathcal F}d_e\ell_e/c_e$ is
at least $\theta k/n-\theta$, since there are at most $n$ required
edges and at least $k$ labels on edges with $d_e\le\theta$.
Discarding the fractional item loses at most another $\theta$ in
demand, and $\theta k/n\ge d(X\setminus B)-\theta$ by rounding.
The edges in $B$, together with those scanned in Step~3, therefore
have total demand at least $d(X)-3\theta$.
If no edge is rejected, this bound holds for $S$.
Otherwise, a rejected edge has demand at most $\theta$, so
$d(S)>1-d_{e_0}-\theta\ge d(X)-\theta$.

The preceding bounds give $c(K)\le(1+2\theta)c(P)$.
Restoring complete paths to original edges preserves cost and connectivity.
Capacity follows from the acceptance rule; $2\theta,3\theta\le\varepsilon_0/4$
give the claims.
\end{proof}

The two endpoint constructions give the bound on $\Phi(R')$
required by Lemma~\ref{lem:selection}.

\begin{lemma}\label{carp:lem:score}
Let $R$ be a feasible route with unique non-small required edge $e_0$ and
$\Delta(R):=p(E(R))-2\rho(R)-c(R)/2>0$.
Write $H=2\rho(R)$. For an available set
$A\subseteq E_{\mathrm{small}}$, let
$U=(E(R)\cap E_{\mathrm{small}})\setminus A$.
The construction in Section~\ref{carp:sec:local} returns a feasible
route $R'$ serving $e_0$ and a subset of $A$ with
\[
\Phi(R')\ge2\Delta(R)-\frac3{10}p(U)+3\varepsilon_0H.
\]
\end{lemma}

\begin{proof}
Write $e_0=ab$. Averaging the two route costs
cancels the endpoint-distance terms.

Write $L=c(R)$, $z=6h_{e_0}/7$, and
$A'=\{e\in A:h_e\ge z\}$.
Using prices, capacity, and $h_e\le H\le L$ for $e\in E(R)$,
the calculation in Lemma~\ref{lem:exceptional-range} gives
$7/16<d_{e_0}<9/16$, $h_{e_0}>0$, and $L<8H/7$.
Deleting the designated traversal of $ab$ from $R$ gives an
$r$--$a$ walk $P_a$ and, after reversing the other part, an
$r$--$b$ walk $P_b$, with $c(P_a)+c(P_b)=L-c_{e_0}$.
Assign each required edge of $E(R)\cap A'$ to the walk containing its
designated traversal. Apply Lemma~\ref{carp:lem:recovery}
to these two sets, each of total demand at most $1-d_{e_0}$.

Let $K_a,K_b$ be the resulting subgraphs and $R'_a,R'_b$ the routes.
Step~4 removes paths of length at least $c(r,a)$ and $c(r,b)$
from doubled $K_a$ and $K_b$, respectively. Each selected required edge
retains a copy for a complete traversal. Hence
\[
\begin{aligned}
c(R'_a)&\le2c(K_a)+c_{e_0}+c(r,b)-c(r,a),\\
c(R'_b)&\le2c(K_b)+c_{e_0}+c(r,a)-c(r,b).
\end{aligned}
\]
Their mean is at most
$(1+\varepsilon_0/4)(L-c_{e_0})+c_{e_0}\le(1+\varepsilon_0/4)L$.
For each endpoint, the selected small edges have total demand at
least the demand assigned to its walk minus $\varepsilon_0/4$.
Since $p_e\ge(8/7)z d_e$ on $A'$ and $\lambda=3/5$, and
$\Phi(R')\ge[\Phi(R'_a)+\Phi(R'_b)]/2$, we obtain
\[
\Phi(R')\ge p_{e_0}-L+\frac{12}{35}z\,d(E(R)\cap A')
-\frac{\varepsilon_0}{4}L-\frac{24}{35}z\,\frac{\varepsilon_0}{4}.
\]
This is the estimate in Lemma~\ref{lem:route-score} before subtracting
$3p(U)/10$. Its remaining cutoff and capacity calculations apply since $p_e=(8/7)d_eh_e$
and $h_e\le H$ for small $e\in E(R)$,
$p_{e_0}\le2h_{e_0}d_{e_0}$, and $L<8H/7$.
\end{proof}

\subsection{Approximation guarantee}\label{carp:sec:guarantee}
Apply Algorithm~\ref{alg:global} with $h_e,p_e,D$ and the singleton routes
above, writing $R_e,S_e,p'_e$ for its stored routes, sets of small required
edges, and adjusted prices. Make these changes:
\begin{itemize}[nosep]
\item Use $R_{\mathrm{RPP}}$ instead of $R_{\mathrm{TSP}}$ and
the candidates from Section~\ref{carp:sec:local} instead of
Algorithm~\ref{alg:quota}.
\item For each enumerated set of two to four required edges, compute a minimum-cost route
by enumerating their orders and service directions, at most $4!\,2^4$ choices.
\item Complete the unserved subset of $E^+$ using Lemma~\ref{carp:lem:completion}.
\end{itemize}
Let $\mathcal R$ be the family after set replacements and
$X=E^+\setminus E(\mathcal R)$.

\begin{proof}[Proof of Theorem~\ref{thm:carp}]
The packing and selection bounds cancel the $2\rho(R)$ terms
in \eqref{carp:eq:full-tour}.

The singleton cases are optimal: if $D=0$, positive demands give
$h_e=0$ for every required edge. For every nonempty stored route
$R_e$, the bound $h_f\le c(R_e)$ for every $f\in S_e$ gives
$p(S_e)\le(8/7)d(S_e)c(R_e)<9c(R_e)/14$.
Lemma~\ref{lem:local}'s proof gives feasible stored routes with
pairwise disjoint sets $E(R_e)$,
$0\le\Phi(R_e)\le s(R_e)\le p_e$, fewer than $n/\varepsilon_0$
updates, and hence $0\le p'_e\le p_e$.

Define $F$ by \eqref{eq:packing-contribution}, replacing $V(R)$ by
$E(R)$ and choosing $S_j(R)$ as the $j$ non-small required edges in
$E(R)$ with largest original demands. Restricting routes does not increase cost, so
Lemma~\ref{lem:packing-bound} applies. The price formula, capacity, and
$h_e\le2\rho(R)\le c(R)$ for $e\in E(R)$ give
Lemma~\ref{lem:exceptional-range}'s bound on $F(R;p)$ and exceptional
demand interval. Thus Lemma~\ref{lem:packing-interface} applies with
$b_2,b_3,b_4$.

Since packing uses only non-small required edges, each selected set
meets $E(R_e)$ only at $e$. Hence Lemma~\ref{lem:replacement} applies.
For an exceptional reference route $R$ with unique non-small required
edge $e$, the stopping loss $\varepsilon_0h_e\le2\varepsilon_0\rho(R)$
and Lemma~\ref{carp:lem:score} give \eqref{eq:stored-score} with $e$ in
place of $v$. Applying Lemma~\ref{lem:selection}'s proof,
\eqref{carp:eq:full-tour}, and Lemma~\ref{carp:lem:completion} gives
\[
\begin{aligned}
c(\mathcal T)
&\le c(R_{\mathrm{RPP}})+c(\mathcal R)+p(X)\\
&\le\sum_{R\in\mathcal T^*}
\left[2c(R)-2\rho(R)+2\rho(R)+\frac12c(R)\right]
=\frac52\opt.
\end{aligned}
\]
The selected routes and completion serve disjoint subsets partitioning $E^+$.

For the running time, $m\le n^2$ and $|B|<1/\theta$ give
$|B|M+m+2=O(n^2/\theta^2)$ vertices in each quota-tree call.
There are $n^{O(1/\theta)}$ choices of $B$; Theorem~\ref{thm:quota-tree}
gives polynomial runtime for fixed $\theta=\varepsilon_0/16$.
\end{proof}

\end{document}